\documentclass[aps,preprintnumbers,nofootinbib,onecolumn]{revtex4}
\usepackage{graphics}
\usepackage{amsfonts,amssymb,amsmath}            
\usepackage{ifthen}
\usepackage{amsthm}
\usepackage{thmtools, thm-restate}
\usepackage{dsfont}
\usepackage{float}
\usepackage{MnSymbol}
\usepackage[caption=false]{subfig}
\usepackage{tikz}

\usepackage[hyperfootnotes=false]{hyperref}
\usepackage{xcolor}
\definecolor{mylinkcolor}{rgb}{0,0,0.8} 
\hypersetup{unicode=true,%
  bookmarksnumbered=false,bookmarksopen=false,bookmarksopenlevel=1, %
  breaklinks=true,pdfborder={0 0 0},colorlinks=true}%
\hypersetup{%
  anchorcolor=mylinkcolor,citecolor=mylinkcolor, %
  filecolor=mylinkcolor,linkcolor=mylinkcolor, %
  menucolor=mylinkcolor,runcolor=mylinkcolor, %
  urlcolor=mylinkcolor}

\usepackage[nodayofweek]{datetime}
\newcommand{\comment}[1]{}
\newcommand{\ket}[1]{| #1 \rangle}
\newcommand{\bra}[1]{\langle #1 |}

\newcommand{\tr}{{\rm tr}}

\newcommand{\cP}{\mathcal{P}}

\newcommand{\Conv}{\mathrm{Conv}}

\newcommand{\chsh}{\mathrm{CHSH}}
\newcommand{\PR}{\mathrm{PR}}
\newcommand{\Cl}{\mathrm{C}}
\newcommand{\BC}{\mathrm{BC}}
\newcommand{\QM}{\mathrm{QM}}
\newcommand{\CG}{\mathrm{CG}}
\newcommand{\RL}{\mathrm{R}}
\newcommand{\E}{\mathrm{E}}
\newcommand{\iso}{\mathrm{iso}}
\newcommand{\mix}{\mathrm{mix}}
\newcommand{\noise}{\mathrm{noise}}

\newcommand{\NL}{\mathrm{NL}}
\newcommand{\Loc}{\mathrm{L}}
\DeclareMathOperator{\Tr}{Tr}

\theoremstyle{plain}
\newtheorem{theorem}{Theorem}
\newtheorem{conjecture}{Conjecture}
\newtheorem{lemma}[theorem]{Lemma}
\newtheorem{corollary}[theorem]{Corollary}
\newtheorem{proposition}{Proposition}

\theoremstyle{definition}
\newtheorem{definition}{Definition}
\newtheorem{remark}{Remark}

\begin{document}

\title{On the insufficiency of entropic inequalities for detecting non-classicality in the Bell causal structure}
\author{V. Vilasini}
\email{vv577@york.ac.uk}
\affiliation{Department of Mathematics, University of York, Heslington, York YO10 5DD.}

\author{Roger Colbeck}
\email{roger.colbeck@york.ac.uk}
\affiliation{Department of Mathematics, University of York, Heslington, York YO10 5DD.}
\date{$17^{\text{th}}$ July 2020}

\begin{abstract}
  Classical and quantum physics impose different constraints on the joint probability distributions of observed variables in a causal structure. These differences mean that certain correlations can be certified as non-classical, which has both foundational and practical importance.  Rather than working with the probability distribution itself, it can instead be convenient to work with the entropies of the observed variables. In the Bell causal structure with two inputs and outputs per party, a technique that uses entropic inequalities is known that can always identify non-classical correlations. Here we consider the analogue of this technique in the generalization of this scenario to more outcomes. We identify a family of non-classical correlations in the Bell scenario with two inputs and three outputs per party whose non-classicality cannot be detected through the direct analogue of the previous technique.  We also show that use of Tsallis entropy instead of Shannon entropy does not help in this case.  Furthermore, we give evidence that natural extensions of the technique also do not help.  More precisely, our evidence suggests that even if we allow the observed correlations to be post-processed according to a natural class of non-classicality non-generating operations, entropic inequalities for either the Shannon or Tsallis entropies cannot detect the non-classicality, and hence that entropic inequalities are generally not sufficient to detect non-classicality in the Bell causal structure.

  In addition, for the bipartite Bell scenario with two inputs and three outputs we find the vertex description of the polytope of non-signalling distributions that satisfy all of the CHSH-type inequalities, which is one of the main regions of investigation in this work.
\end{abstract}

\maketitle

\section{Introduction}
\label{sec:introduction}
Causal structures are a useful tool for understanding correlations between observed events.  Such correlations may be mediated by an influence travelling from one to the other, or come about due to common causes, which may not be observed.  The nature of any unobserved causes depends on the theory being considered.  For instance they may be classical, quantum or from a generalized probabilistic theory (GPT)~\cite{Barrett07}, and the kinds of observed correlations that are possible in general depends on this.  At the foundational level, studying the differences gives us insight into how the notion of causality differs between theories, while, on a practical level, these differences are crucial for applications in device-independent cryptography~\cite{Ekert91,MayersYao,BHK,RogerThesis,CK,Pironio2009}.

One way to establish a difference is to violate a \emph{Bell inequality}~\cite{Bell}, where we use the term to mean a necessary condition on the observed correlations when any unobserved systems are classical. Bell inequalities are often introduced using the (bipartite) Bell causal structure (see Figure~\ref{fig: Bell}(a)). Here there are four observed variables: $A$ and $B$ corresponding to the inputs of each party, and $X$ and $Y$ corresponding to the outputs.  In the case that the numbers of possible inputs are $i_A$ and $i_B$ and likewise the number of possible outputs are $o_A$ and $o_B$, we call the scenario the $(i_A,i_B,o_A,o_B)$ Bell scenario.  For the $(2,2,2,2)$ case, the CHSH inequalities~\cite{CHSH} are known to be the only class of Bell inequalities required to completely characterize the scenario (i.e., all extremal $2$-setting, $2$-outcome Bell inequalities are equivalent to the CHSH inequalities up to symmetry).  In the $(2,2,3,3)$ scenario, which will be the main focus of this paper, there is only one new class of Bell inequalities inequivalent to CHSH, the $I_{2233}$ class~\cite{Kaszlikowski2002, CGLMP02}. In other words, given a no-signalling distribution for the $(2,2,3,3)$ scenario, the distribution is local if and only if all the CHSH and $I_{2233}$ inequalities hold.

As $i_A$, $i_B$, $o_A$ and $o_B$ increase, many new classes of extremal Bell inequalities are found and these scenarios quickly become difficult to fully characterize~\cite{Pitowski89, Masanes2002, Bancal2010, Cope2019}. One attempt at avoiding this difficulty is to move away from probability space to instead consider inequalities expressed in terms of the entropies of the variables involved. There are two ways that these can be used: either directly using the causal structure under consideration, or by using the post-selection technique in which the original causal structure is first modified (more details can be found later in this paper). Braunstein and Caves~\cite{BraunsteinCaves88} were the first to derive an entropic Bell inequality. They considered the post-selected version of the Bell causal structure shown in Figure~\ref{fig: Bell}(b) and found entropic inequalities that hold for all classical distributions. These can be violated when one or more of the unobserved nodes are quantum, and hence behave like entropic versions of Bell inequalities.

Although entropic inequalities can be useful as a way to detect non-classicality, the entropic approach has several disadvantages. For instance, in the bipartite Bell causal structure without post-selection the set of achievable Shannon entropies over the observed variables for classical and quantum causes coincide~\cite{Weilenmann16}, so without post-selecting, non-classicality cannot be detected through entropic Bell inequalities in this case. The use of other entropic measures such as Tsallis entropies to analyse this problem in the absence of post-selection has also been shown to have limitations~\cite{Vilasini2019} and no quantum violations are known for the entropic Bell inequalities derived in~\cite{Vilasini2019}.  Because of this, we focus on post-selected causal structures in this paper.

\begin{figure}[t!]
\centering
	\subfloat[]{\includegraphics[scale=1.0]{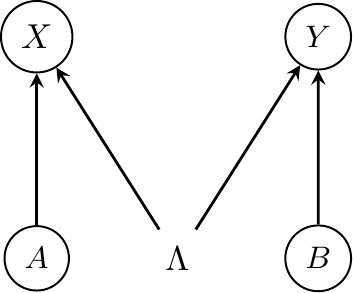}}\qquad\qquad\qquad\qquad\qquad\qquad
\subfloat[]{\includegraphics[scale=1.0]{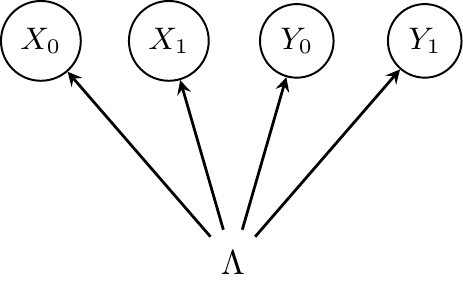}}
\caption{(a) The bipartite Bell causal structure. The nodes $A$ and
  $B$ represent the random variables corresponding to independently
  chosen inputs, while $X$ and $Y$ represent the random
  variables corresponding to the outputs. $\Lambda$ is an unobserved
  node representing the common cause of $X$ and $Y$. (b) The
  post-selected Bell causal structure for two parties. The observed
  nodes $X_a$ represent the outputs when the input is $a\in\{0,1\}$
  and likewise for $Y_b$. Note that $X_0$ and $X_1$ are never
  simultaneously observed and likewise $Y_0$ and $Y_1$.}\label{fig: Bell}
\end{figure}

It is natural to ask whether the non-classicality of a distribution can always be detected through post-selected entropic inequalities. For the $(d,d,2,2)$ Bell scenarios with $d\geq 2$, this is known to be the case~\cite{Chaves13} in the following sense.  For every non-classical distribution in the $(d,d,2,2)$ Bell scenario, there is a transformation that does not make any non-classical distribution classical, and such that the resulting distribution violates one of the BC entropic inequalities. The main purpose of this work is to investigate whether a similar result holds for non-binary outcomes. To do so, we need to specify a class of post-processing operations. The most general operations that we could consider are the \emph{non-classicality non-generating (NCNG)} operations, i.e., those that do not map any classical distribution to a non-classical one.  An interesting subset of these is the class of post-processings achievable through \emph{local operations and shared randomness (LOSR)}, which are physical in the sense that two separated parties with shared randomness could perform them\footnote{In general NCNG operations used on the correlations prior to evaluating an entropic inequality need not be physical in this sense.}. Because of the difficulty of dealing with arbitrary NCNG operations, for the majority of our analysis we consider LOSR supplemented with the additional (NCNG) operation where the parties are exchanged (and convex combinations). We use LOSR+E to refer to this supplemented set.

We study the $(2,2,3,3)$ Bell scenario with LOSR+E post-processing operations, to see whether when applied to any non-classical distribution the result violates an entropic Bell inequality. We investigate this using both Shannon and Tsallis entropies. Our motivation for considering Tsallis entropies is that they are known to provide an advantage over the Shannon entropy in detecting non-classicality in the absence of post-processing~\cite{Wajs15} in the sense that there are non-classical distributions that violate Tsallis entropic inequalities but not the analogous Shannon-entropic inequality. In the $(2,2,2,2)$ case however, due to the result of \cite{Chaves13}, this advantage is less apparent when post-processings are considered. It is unclear whether or not this is also the case for the $(2,2,3,3)$ scenario, and hence we consider Tsallis entropies in this work.

The result of~\cite{Chaves13} that Shannon inequalities can always be used in the case that CHSH is violated readily extends to the $(2,2,3,3)$ scenario (cf.\ Corollary~\ref{coroll:chaves}). Thus, of most interest to us is the region containing non-classical distributions that satisfy all the CHSH-type inequalities.  After finding the vertex description of this region, we show that some distributions in this region violate entropic inequalities and conjecture that there are others that cannot violate either Shannon or Tsallis entropic inequalities after processing with LOSR+E. For our conjecture, we consider a class of isotropic non-classical distributions in the $(2,2,3,3)$ scenario and give numerical evidence that arbitrary Shannon entropic inequalities and a class of Tsallis entropic inequalities cannot detect the non-classicality of these distributions under any processing with LOSR+E. If correct, our conjecture implies that the method of~\cite{Chaves13} for the $(2,2,2,2)$ scenario does not work for all non-classical distributions in the $(2,2,3,3)$ scenario. That said, it remains a useful technique in many cases.

The structure of the remainder of the paper is as follows. After introducing our notation and reviewing some existing work in Section~\ref{sec: prelim}, we proceed to investigate the $(2,2,3,3)$ scenario. We present new results for this scenario in probability space (Section~\ref{sec: 2233prob}) as well as in entropy space (Section~\ref{sec: 2233ent}). Finally, in Section~\ref{sec: disc&conc}, we conclude and discuss some open questions.  These results, along with those of~\cite{Weilenmann16, Vilasini2019}, highlight some of the limitations of the entropic approach to analysing causal structures.

\section{Preliminaries}
\label{sec: prelim}
\subsection{Probability distributions and entropy}
\label{ssec: prob_ent}
We begin with some notation.  Given a conditional probability
distribution $p_{XY|AB}$ where $A$, $B$, $X$ and $Y$ have
cardinalities $i_A$, $i_B$, $o_A$ and $o_B$ respectively, we can
express the distribution using a matrix.  For instance, in the case
where all the variables take values in $\{0,1\}$ and using $p(xy|ab)$
as an abbreviation for $p_{XY|AB}(xy|ab)$, this is done as
\begin{equation}
\label{eq: dist}
    p_{XY|AB}=
\begin{array}{ |c|c|} 
 \hline
 p(00|00) \quad p(01|00) & p(00|01) \quad p(01|01)\\ 
 p(10|00) \quad p(11|00) & p(10|01) \quad p(11|01)\\ 
 \hline
 p(00|10) \quad p(01|10) & p(00|11) \quad p(01|11)\\ 
 p(10|10) \quad p(11|10) & p(10|11) \quad p(11|11)\\ 
 \hline
\end{array}
\end{equation}
and the generalisation to larger alphabets is analogous (see,
e.g.,~\cite{Cirelson93}).  This format is convenient because it makes
it easy to check whether a distribution is no-signalling, i.e., to check
that $p_{X|AB}$ is independent of $B$ and that $p_{Y|AB}$ is
independent of $A$.

In a given scenario we will be interested in the set of distributions that are non-signalling, and the set of classical (local) distributions. These sets form convex polytopes that are highly symmetric. In particular, such polytopes are invariant under local relabellings and/or relabelling parties. By \emph{local relabellings} we mean combinations of relabelling the inputs (e.g., $A\mapsto A\oplus 1$) and outputs conditioned on the local input (e.g., $X\mapsto X\oplus\alpha A\oplus\beta$ where $\alpha, \beta \in \{0,1\}$ and $\oplus$ denotes modulo-2 addition). One might also think about more general 
\emph{global relabellings} that depend on both inputs (for instance maps of the form $X\mapsto X\oplus\alpha A\oplus\beta B\oplus\gamma$ with $\alpha, \beta, \gamma \in \{0,1\}$), but these do not preserve the no-signalling set in general so will not be considered here. The only global relabelling we consider is exchange of the two parties, which corresponds to transposing the distribution in the matrix notation of Equation~\eqref{eq: dist} and always preserves no-signalling.

Given a random variable $X$, distributed according to the discrete probability distribution $p_x:=p_X(x)$ the Shannon entropy is defined by\footnote{Note that we choose to use
  the natural logarithm in this work, in contrast to the usual choice
  of using base 2, in order to make the relation to the Tsallis
  entropy more straightforward.}
$H(X)=-\sum_{\{x:p_x>0\}} p_x\ln p_x$.  Given two random variables, $X$ and
$Y$, the conditional Shannon entropy is defined by
$H(X|Y)=-\sum_{\{xy:p_{xy}>0\}}p_{xy}\ln p_{x|y}$. The following
properties hold:
\begin{itemize}
    \item \textbf{P1} Monotonicity: $H(X)\leq H(XY)$.
    \item \textbf{P2} Strong-subadditivity: $H(XY)+H(YZ)\geq H(XYZ)+H(Y)$.
    \item \textbf{P3} Chain rule: $H(X|Y)=H(XY)-H(Y)$.
\end{itemize}

The order $q$ Tsallis entropy of $X$ for a real parameter $q$ is defined as~\cite{Tsallis1988}
\begin{equation}
\label{eq: tsallis1}
S_q(X)=\begin{cases}-\sum_{\{x:p_x>0\}}p_x^q\ln_q p_x &q\neq 1
  \\H(X)&q=1\end{cases}\,.
\end{equation}
In this expression we have used the $q$-logarithm function $\ln_q p_x=\frac{p_x^{1-q}-1}{1-q}$, which converges to the natural logarithm function as $q\to 1$.  This means that $\lim\limits_{q\rightarrow 1}S_q(X)=H(X)$ and hence that $S_q(X)$ is continuous in $q$. We henceforth write $\sum_x$ instead of $\sum_{\{x:p_x>0\}}$ with the implicit understanding that probability zero events are excluded from the sum.

The Tsallis entropies for $q\geq1$ satisfy many of the same properties as the Shannon entropy.  In particular, monotonicity, strong subadditivity and chain rule all hold for Tsallis entropies for all $q\geq1$~\cite{Daroczy1970, Furuichi04}, making these polymatroids like the Shannon entropy.

\subsection{Causal structures}\label{ssec: PS}
The relationships between different variables of interest can be conveniently expressed as a causal structure.  This is a \emph{directed acyclic graph} (DAG) where the observed variables are nodes, and there may be additional nodes representing unobserved systems. Given such a causal structure, we distinguish the cases where the hidden systems are classical, quantum or are from some generalized probabilistic theory. For every classical causal structure that has at least one parentless observed node, a post-selected causal structure can be defined. The general technique for doing this can be found in~\cite{Weilenmann20170483} (for example).

In this work we will only consider the Bell causal structure with two inputs per party and the post-selected version thereof (see Figure~\ref{fig: Bell}). The post-selected causal structure is obtained by removing the parentless observed nodes $A$ and $B$ in the original causal structure~\ref{fig: Bell}(a) and replacing the descendants $X$ and $Y$ with two copies of each i.e., $X_{A=0}$, $X_{A=1}$, $Y_{B=0}$, $Y_{B=0}$ such that the original causal relations are preserved and there is no mixing between the copies (this is shown in Figure~\ref{fig: Bell}(b)).  It makes sense to do this in the classical case because classical information can be copied, so we can simultaneously consider the outcome $X$ given $A=0$ and that given $A=1$.  By contrast, in the quantum case the values of $A$ correspond to different measurements that are used to generate $X$, and the associated variables $X_{A=0}$ and $X_{A=1}$ may not co-exist. It hence does not make sense to consider a joint distribution over $X_{A=0}$ and $X_{A=1}$ in this case.  We therefore only consider the subsets of the observed variables that co-exist
\begin{equation}
\label{eq: coexisting}
 \mathcal{S}:=\{X_0,\: X_1,\: Y_0,\: Y_1,\: X_0Y_0,\: X_0Y_1,\: X_1Y_0,\: X_1Y_1\} ,  
\end{equation}
where we use the short form $X_0Y_0$ for the set $\{X_0,Y_0\}$ etc.
Any non-trivial inequalities derived for the co-existing sets in the classical case can admit quantum or GPT violations.

\subsection{The (2,2,2,2) Bell scenario in probability space}
\label{ssec: 2222prob}

For the bipartite Bell causal structure of Figure~\ref{fig: Bell}(a), the set of all observed distributions $p_{XY|AB}$ that can arise when $\Lambda$ is classical corresponds to the set of correlations that admit a local hidden variable model, i.e., the set of distributions that have the form
 \begin{equation}
 \label{eq: loccaus}
 p_{XY|AB}=\int_{\Lambda}\mathrm{d}\Lambda\, p_{\Lambda}\, p_{A}\, p_{B}\, p_{X|A\Lambda}\, p_{Y|B\Lambda}.
 \end{equation}
 In this work we will refer to such correlations either as \emph{local} or as \emph{classical} and denote the set of all such distributions $\mathcal{L}$. We also use $\mathcal{L}^{(2,2,2,2)}$ to denote the local distributions in the $(2,2,2,2)$ case (and analogously for other cases).
 
 The set of local correlations form a convex polytope, which can be specified in terms of a finite set of Bell inequalities, each a necessary condition for classicality.  In the $(2,2,2,2)$ case, there are eight extremal Bell inequalities (facets of the local polytope).  One has the form
\begin{equation}
\label{eq: CHSH}
    I_{\chsh}:=p(X=Y|A=0,B=0)+p(X=Y|A=0,B=1)+p(X=Y|A=1,B=0)+p(X\neq
    Y|A=1,B=1)\leq 3
\end{equation}
and the other seven are equivalent under local
relabellings~\cite{Cirelson93}. We denote these by $I_{\chsh}^k$ for
$k\in [8]$, where $I_{\chsh}^1:=I_{\chsh}$ and $[n]$ stands for the
set $\{1,2,....,n\}$ where $n$ is a positive integer. This provides
the facet description of the $(2,2,2,2)$ local polytope.  One can also express $I_{\chsh}$ in matrix form using
\begin{equation}\label{eq:mat1}
    M_{\chsh}=\begin{array}{|cc|cc|}\hline 1&0&1&0\\0&1&0&1\\\hline1&0&0&1\\0&1&1&0\\\hline\end{array}\,,
\end{equation}
so that the Bell inequality can be written $\tr\left(M_{\chsh}^T P\right)\leq3$, where $P$ is the matrix form of the distribution and $T$ denotes the transpose.

In the vertex picture, the $(2,2,2,2)$ local polytope has 16 local deterministic vertices and the $(2,2,2,2)$ non-signalling polytope shares the vertices of the local polytope and has eight more: the Popescu-Rohrlich (PR) box and seven distinct local relabellings~\cite{Cirelson93,PR}. The PR box distribution satisfies $X\oplus Y=A.B$ and has the form
\begin{equation}
\label{eq: PR}
    p_{\PR}=
\begin{array}{ |c|c|} 
 \hline
 \frac{1}{2} \quad 0 & \frac{1}{2} \quad 0\\ 
 0 \quad \vphantom{\frac{1}{f}}\frac{1}{2} & 0 \quad \vphantom{\frac{1}{f}}\frac{1}{2}\\
 \hline
 \frac{1}{2} \quad 0 & 0 \quad \frac{1}{2}\\ 
 0 \quad \vphantom{\frac{1}{f}}\frac{1}{2} & \vphantom{\frac{1}{f}}\frac{1}{2} \quad 0\\ 
 \hline
\end{array}\,.
\end{equation}
We denote the eight extremal non-signalling vertices equivalent under local relabellings to $p_{\PR}$ by $p_{\PR}^k$, $k\in [8]$ where $p_{\PR}^1:=p_{\PR}$. Note that the 8 CHSH inequalities $\{I_{\chsh}^k\}$ are in one-to-one correspondence with these 8 extremal non-signalling points i.e., each $p_{\PR}^k$ violates exactly one CHSH inequality and each CHSH inequality is violated by exactly one $p_{\PR}^k$.

\subsection{Entropic inequalities and post-selection}
\label{ssec: entineq}
In \cite{BraunsteinCaves88}, Braunstein and Caves derived a set of
constraints on the post-selected causal structure of Figure~\ref{fig:
  Bell}(b) and showed that these constraints can be violated by
quantum correlations.  To discuss these we introduce the notion
of \emph{entropic classicality}. For every distribution $p_{XY|AB}$ in the Bell causal structure (Figure~\ref{fig: Bell}(a)), we can associate an entropy vector $v\in\mathbb{R}^{8}$ in the post-selected causal structure (Figure~\ref{fig: Bell}(b)) whose components are the entropies of each element of the set $\mathcal{S}$ (Equation~\eqref{eq: coexisting}) distributed according to $p_{X_aY_b}:=p_{XY|A=a,B=b}$. Let $\mathbf{H}$ be
the map that takes the observed distribution to
its corresponding entropy vector in the post-selected causal structure. 

\begin{definition}[Entropic classicality]
\label{definition: entclass}
An entropy vector $v\in\mathbb{R}^8$ is \emph{classical} with respect to the bipartite Bell causal structure (Figure~\ref{fig: Bell}(a)) if there exists a classical distribution $p_{XY|AB}\in\mathcal{L}$ such that $\mathbf{H}(p_{XY|AB})=v$. Further, a distribution $p_{XY|AB}$ is \emph{entropically classical} if there exists a classical distribution with the same entropy vector, i.e., if there exists a classical entropy vector $v$ such that $\mathbf{H}(p_{XY|AB})=v$.
\end{definition}
 
The set of all classical entropy vectors forms a convex cone. The distribution $p_{\PR}$ (Equation~\eqref{eq: PR}) is an example of a nonclassical distribution that is entropically classical (see Section~\ref{ssec: 2222}).

We now review how the Braunstein Caves (BC) Inequalities are derived for the case when the observed parentless nodes $A$ and $B$ are binary. In this case, the post-selected causal structure~\ref{fig: Bell}(b) imposes no additional constraints on the distribution (or entropies) of the observed nodes $X_0$, $X_1$, $Y_0$ and $Y_1$ because they share a common parent and thus any joint distribution over $X_0$, $X_1$, $Y_0$ and $Y_1$ can be realised in the causal structure~\ref{fig: Bell}(b). By contrast, any correlations in the original causal structure~\ref{fig: Bell}(a) must obey the no-signalling constraints over the observed nodes $A$, $B$, $X$ and $Y$ since $A$ does not influence $Y$ and $B$ does not influence $X$ in this causal structure. The inequalities derived by Braunstein and Caves follow by applying Properties~\textbf{P1}-\textbf{3} to the variables $\{X_0,X_1,Y_0,Y_1\}$. The derived relations hold for the classical causal structure (and not necessarily for the quantum and GPT cases) because only in the classical case does it make sense to consider a joint distribution over these four variables that in the quantum and GPT cases do not co-exist (cf.\ Section~\ref{ssec: PS}). It is worth remarking that without post-selection, no quantum-violatable entropic constraints exist for this causal structure~\cite{Weilenmann16}. The BC inequalities are entropic Bell inequalities i.e., they hold for every classical entropy vector in the post-selected causal structure~\ref{fig: Bell}(b).  There are four BC inequalities
\begin{widetext}
\begin{equation}
\label{eq: BCineqs}
    \begin{split}
        I_{\BC}^1:=H(X_0Y_0)+H(X_1)+H(Y_1)-H(X_0Y_1)-H(X_1Y_0)-H(X_1Y_1)\leq 0\\
        I_{\BC}^2:=H(X_0Y_1)+H(X_1)+H(Y_0)-H(X_0Y_0)-H(X_1Y_0)-H(X_1Y_1)\leq 0\\
        I_{\BC}^3:=H(X_1Y_0)+H(X_0)+H(Y_1)-H(X_0Y_0)-H(X_0Y_1)-H(X_1Y_1)\leq 0\\
        I_{\BC}^4:=H(X_1Y_1)+H(X_0)+H(Y_0)-H(X_0Y_0)-H(X_0Y_1)-H(X_1Y_0)\leq 0\\
    \end{split}
\end{equation}
\end{widetext}
It has been shown in~\cite{Fritz13} that these four
inequalities are complete in the following sense (the lemma below is
implied by Corollary~V.3 in~\cite{Fritz13}).
\begin{lemma}\label{lem:BCcomplete}
  A distribution in the postselected Bell scenario with binary $A$ and $B$ is entropically classical if and only if it satisfies the four BC inequalities~\eqref{eq: BCineqs}.
\end{lemma}
It turns out that in the $(2,2,2,2)$ Bell scenario, non-classical distributions that do not violate the BC inequalities can be made to do so
with some additional post-processing, as shown
in~\cite{Chaves13}. We review this result below before analysing the
same question in the $(2,2,3,3)$ scenario.

\subsection{Detecting non-classicality in the (2,2,2,2) Bell scenario in entropy space}\label{ssec: 2222}
 
The current section summarises the relevant results of~\cite{Chaves13}
regarding the sufficiency of entropic inequalities in the $(2,2,2,2)$
scenario. As previously mentioned, it is possible for a non-classical
distribution to have the same entropy vector as a classical one and hence to
be entropically classical. For example, the maximally non-classical
distribution in probability space, $p_{\PR}$ (Equation~\eqref{eq: PR})
is entropically classical since it has the same entropy vector as the
classical distribution
\begin{equation}
    \label{eq: Pc}
     p_{\Cl}=
\begin{array}{ |c|c|} 
 \hline
 \frac{1}{2} \quad 0 & \frac{1}{2} \quad 0\\ 
 0 \quad \vphantom{\frac{1}{f}}\frac{1}{2} & 0 \quad \vphantom{\frac{1}{f}}\frac{1}{2}\\
 \hline
 \frac{1}{2} \quad 0 & \frac{1}{2} \quad 0\\ 
 0 \quad \vphantom{\frac{1}{f}}\frac{1}{2} & 0 \quad \vphantom{\frac{1}{f}}\frac{1}{2}\\
 \hline
\end{array}
\end{equation}
and hence cannot violate any of the BC inequalities\footnote{$p_{\PR}$ and $p_{\Cl}$
  are related by a permutation of the entries in the bottom right
  $2\times 2$ block and entropies are invariant under such
  permutations.}. However, the distribution
$\frac{1}{2}p_{\PR}+\frac{1}{2}p_{\Cl}$ maximally violates $I_{\BC}^4\leq0$
attaining a value of $\ln 2$. That convex mixtures of non-violating distributions can lead to a violation is due to the
fact that entropic inequalities are non-linear in the underlying
probabilities (in contrast to the facet Bell inequalities in
probability space).

In~\cite{Chaves13}, it was shown that such a procedure is possible for every non-classical distribution in the $(d,d,2,2)$ Bell scenario with $d\geq 2$ i.e., for every such distribution, there exists an LOSR transformation such that the resultant distribution violates a Shannon entropic BC inequality~(\ref{eq: BCineqs}).  Thus, non-classicality can be detected in this scenario by processing the observed correlations (in a way that cannot generate non-classicality) before using a BC entropic inequality on the result.  In this sense the BC entropic inequalities provide a necessary and sufficient test for non-classicality in these scenarios.

In more detail, for the $(2,2,2,2)$ case this works as follows.  First, one defines a special class of distributions, \emph{isotropic distributions}, as follows for some $k\in [8]$ and $\epsilon \in [0,1]$.
\begin{equation}
   p_{\iso}^k=\epsilon p_{\PR}^k+(1-\epsilon)p_{\noise},
\end{equation}
where $p_{\noise}$ is white noise i.e., the distribution with all entries equal to $1/4$. In the $(2,2,2,2)$ Bell scenario the isotropic distribution $p_{\iso}^k$ is non-classical if and only if $\epsilon>1/2$. The LOSR transformation used in~\cite{Chaves13} involves first transforming the observed distribution into an isotropic distribution through a local depolarisation procedure that cannot generate non-classicality. Second, it is shown that for any non-classical isotropic distribution i.e., a $p_{\iso}^k$ with $\epsilon>1/2$, there exists a classical distribution $p_{\Cl}^k$ such that the distribution $p_{\E,v}^k=vp_{\iso}^k+(1-v)p_{\Cl}^k$ violates one of the BC entropic inequalities for sufficiently small $v>0$. In particular, the value of $I_{\BC}^k$ for $p_{\E,v}^k$ can be expanded for small $v$ as
\begin{equation}
\label{eq: BCexpand2}
    I_{\BC}^k \approx \frac{v}{\ln 4}[f(\epsilon)-(4\epsilon-2)\ln v],
\end{equation}
where $f(\epsilon)$ is a function of $\epsilon$, independent of $v$
(see~\cite{Chaves13} for details). Thus for any $\epsilon>1/2$, the
corresponding isotropic distributions are non-classical and taking $v$
arbitrarily small can make $I_{\BC}^k$ positive which is a violation
of the entropic inequality. We summarise the main result
of~\cite{Chaves13} for $(2,2,2,2)$ Bell scenarios in the following
Theorem (which is implicit in~\cite{Chaves13}).
\begin{theorem}
\label{theorem: chaves}
For every non-classical distribution, $p_{XY|AB}$ in the $(2,2,2,2)$ Bell scenario, there exists an LOSR transformation $\mathcal{T}$, such that $\mathcal{T}(p_{XY|AB})$ violates one of the BC entropic inequalities~\eqref{eq: BCineqs}.
\end{theorem}

One of the aims of the present paper is to study whether this result
extends to the case where the number of outcomes per party is more
than two. In general, the $(2,2,d,d)$ Bell polytope for $d>2$ has new,
distinct classes of Bell inequalities and extremal non-signalling
vertices other than the CHSH inequalities and the PR boxes. In the
following, we analyse this problem for the $d=3$ case, for which it is
helpful to first describe the $(2,2,3,3)$ scenario in probability space.

\subsection{The (2,2,3,3) Bell scenario in probability space}
\label{ssec: 2233rev}
In the $(2,2,3,3)$ Bell scenario, there are two classes of Bell inequalities that completely characterize the local polytope: the CHSH inequalities and the $I_{2233}$ inequalities~\cite{Kaszlikowski2002, CGLMP02} of which a representative example is
\begin{widetext}
\begin{equation}
\label{eq: CGLMP}
\begin{split}
    I_{2233}:=\ &[p(X=Y|A=0,B=1)+p(X=Y-1|A=1,B=1)+p(X=Y|A=1,B=0)+p(X=Y|A=0,B=0)]\\
    &-[p(X=Y-1|X=0,B=1)+p(X=Y|A=1,B=1)+p(X=Y-1|A=1,B=0)\\
    &\ \ \ +p(X=Y+1|A=0,B=0)]\leq 2\,,
\end{split}
\end{equation}
\end{widetext}
where all the random variables take values in $\{0,1,2\}$ and all additions and subtractions of the random variables are modulo 3.  In matrix form a representative CHSH-type inequality and $I_{2233}$ are
\begin{equation}\label{eq:mat2}
    M_{\chsh}^{(2,2,3,3)}=\begin{array}{|ccc|ccc|}\hline 1&0&0&1&0&0\\0&1&1&0&1&1\\0&1&1&0&1&1\\\hline1&0&0&0&1&1\\0&1&1&1&0&0\\0&1&1&1&0&0\\\hline\end{array}\quad\text{and}\quad
    M_{I_{2233}}=\begin{array}{|ccc|ccc|}\hline 1&0&-1&1&-1&0\\-1&1&0&0&1&-1\\0&-1&1&-1&0&1\\\hline1&-1&0&-1&1&0\\0&1&-1&0&-1&1\\-1&0&1&1&0&-1\\\hline\end{array}\,.
\end{equation}

The $(2,2,3,3)$ local polytope has a total of 1116 facets, $36$ of which correspond to positivity constraints, $648$ to CHSH facets (these are equivalent to first coarse-graining two of the outputs into one (for each party and each input) and then applying one of the eight (2,2,2,2) CHSH inequalities\footnote{For instance, evaluating the CHSH-type inequality represented by $M_{\chsh}^{(2,2,3,3)}$ is equivalent to coarse graining the distribution by always mapping outcomes 1 and 2 to 1 and then evaluating $M_{\chsh}$~\eqref{eq:mat1}.}), and the remaining $432$ are $I_{2233}$-type~\cite{CollinsGisin04} (we label these $I_{2233}^i$ for $i\in\{1,2,\ldots,432\}$ with $I_{2233}^1=I_{2233}$).

The facets of the no-signalling polytope correspond to positivity constraints. Converting this facet description to the vertex description (e.g., using the {\sc Porta} software~\cite{porta}) one can obtain all the vertices of the $(2,2,3,3)$ non-signalling polytope.  This comprises $81$ local deterministic vertices, $648$ PR-box type vertices and $432$ extremal non-signalling vertices (for each of the $I_{2233}$ inequalities there is one of the latter that gives maximal violation).  We call these new vertices the \emph{$I_{2233}$-vertices}.  The specific vertex that maximally violates~\eqref{eq: CGLMP} is
\begin{equation}
\label{eq: NL3}
p_{\NL}:=
    \begin{array}{ |c|c|} 
 \hline
 \frac{1}{3} \quad 0 \quad 0 & \frac{1}{3} \quad 0 \quad 0\\ 
 0 \quad \frac{1}{3} \quad 0 & 0 \quad \frac{1}{3} \quad 0\\
 0 \quad 0 \quad \vphantom{\frac{1}{f}}\frac{1}{3} & 0 \quad 0 \quad \vphantom{\frac{1}{f}}\frac{1}{3}\\
 \hline
 \frac{1}{3} \quad 0 \quad 0 & 0 \quad \frac{1}{3} \quad 0\\ 
 0 \quad \frac{1}{3} \quad 0 & 0 \quad 0 \quad \frac{1}{3}\\
 0 \quad 0 \quad \vphantom{\frac{1}{f}}\frac{1}{3} & \vphantom{\frac{1}{f}}\frac{1}{3} \quad 0 \quad 0\\
 \hline
\end{array}
\end{equation}

The $432$ $I_{2233}$ vertices of the $(2,2,3,3)$ non-signalling polytope are related to each other through local relabellings\footnote{In general, equivalent points of the non-signalling polytope may be related by local relabellings or exchange of the two parties (which is a global operation). In the $(2,2,3,3)$ scenario there are $2\times(2\times 6^2)^2=10368$ such operations, twice the number of local relabellings. To count these, note that for each party there are $2$ ways to permute the inputs, and $6$ ways to permute the outputs for each of the $2$ inputs. All 432 extremal points of the $I_{2233}$ type (those which maximally violate a $I_{2233}$ inequality) can be generated using only local relabellings of $p_{\NL}$, and similarly all 648 extremal points of the CHSH type can be generated through local relabellings of $p_{\PR}$ embedded in the $(2,2,3,3)$ scenario (by adding zero probabilities to the third outcome).}.

\section{Results 1: The \texorpdfstring{$(2,2,3,3)$}{(2,2,3,3)} scenario in probability space}\label{sec: 2233prob}
In this section we compute the vertex description of the \emph{CHSH-classical polytope}, $\Pi_{\chsh}^{(2,2,3,3)}$, i.e., the polytope whose facets are the 648 CHSH inequalities and the positivity constraints.  As previously mentioned, this will be the main region of interest in the remainder of this work since the non-classicality of distributions not belonging to this region can always be certified using Shannon-entropic inequalities (Corollary~\ref{coroll:chaves}). 
The following result allows us to significantly speed up the vertex enumeration problem.

\begin{proposition}
\label{proposition: probspace}
Every non-classical distribution in $\Pi_{\chsh}^{(2,2,3,3)}$ violates only one $I_{2233}$ inequality. 
\end{proposition}
\begin{proof}
 Let $i\in\{2,3,\ldots,432\}$ and consider the linear program that maximises the value of $\epsilon \geq 0$ subject to there existing a no-signalling distribution that
 \begin{itemize}
     \item violates $I_{2233}^1\leq2$ and $I_{2233}^i\leq2$ by at least $\epsilon$, i.e., $I_{2233}^1-\epsilon \geq 2$ and $I_{2233}^i-\epsilon \geq 2$;
     \item satisfies all CHSH-type inequalities.
 \end{itemize}
 We run over $i\in\{2,3,\ldots,432\}$ and check that in all cases either the output of this linear program is $\epsilon=0$ (meaning that the two $I_{2233}$ inequalities can be jointly saturated but not violated) or that the program is infeasible (the two $I_{2233}$ inequalities cannot even be jointly saturated). By symmetry it follows that no pair of $I_{2233}$ inequalities can be simultaneously violated when all the CHSH-type inequalities are satisfied.
\end{proof}

Note that in the $(2,2,3,3)$ scenario there exist extremal non-signalling distributions that violate multiple Bell inequalities. For example, the distributions $p_{\NL}$ (Equation~\eqref{eq: NL3}) and $p_{\NL}^*$ (Equation~\eqref{eq: NL*}) violate $I_{2233}\leq2$ (cf.~\eqref{eq: CGLMP}) although only $p_{\NL}$ violates it maximally. By symmetry, $p_{\NL}$ also violates another $I_{2233}$ inequality. 
In addition, $p_{\NL}$ violates the CHSH-type inequality whose evaluation is equivalent to applying the output coarse-graining $0\mapsto0$, $1\mapsto1$ and $2\mapsto1$ for each party and then evaluating~\eqref{eq: CHSH}. This is in contrast to the $(2,2,2,2)$ scenario where there is a one-to-one correspondence between the extremal non-signalling vertices and the CHSH inequalities in the sense that each such vertex violates exactly one CHSH inequality.\footnote{Note that this correspondence breaks down in the $(2,2,3,3)$ scenario where it is possible for a CHSH-type vertex to violate multiple CHSH-type inequalities (these correspond to the same 2-outcome CHSH inequality after coarse-graining).}

Due to the symmetries, all the vertices of $\Pi_{\chsh}^{(2,2,3,3)}$ can be enumerated by first finding all the vertices for which the $I_{2233}$ inequality of Equation~\eqref{eq: CGLMP} is saturated or violated i.e., $I_{2233}\geq 2$, and taking the orbit of these vertices under local relabellings and exchange of the parties. The vertex enumeration for this case yields 47 extremal points of which 30 are the local deterministic points that saturate $I_{2233}\geq2$ and 17 are non-classical points that violate only this inequality. These are listed in Table~\ref{tab:vertices}. By applying all symmetries and removing duplicate vertices we find that $\Pi_{\chsh}^{(2,2,3,3)}$ has 7425 vertices (including the 81 local deterministic vertices).

\section{Results 2: The \texorpdfstring{$(2,2,3,3)$}{(2,2,3,3)} scenario in entropy space}
\label{sec: 2233ent}

We now investigate whether entropic inequalities are necessary and sufficient for non-classicality in the $(2,2,3,3)$ Bell scenario. In $(d,d,2,2)$ scenarios with $d\geq 2$, only $2d$ Shannon entropic inequalities are required for the Shannon entropic characterisation of the scenario~\cite{Chaves2012, Fritz13}; in the $(2,2,2,2)$, these are the four inequalities of~\eqref{eq: BCineqs}. It may at first seem surprising that these can always be used to decide whether a distribution is classical because the number of extremal Bell inequalities grows very rapidly in $d$ in the $(d,d,2,2)$ scenario~\cite{Cope2019}, and deciding whether a distribution is classical is NP-complete~\cite{Avis04}. The reduction in the number of inequalities in entropy space is compensated by the need to identify a suitable post-processing operation (of which there are uncountably many possibilities) in order to detect violations. 

The first observation is a corollary of Theorem~\ref{theorem: chaves}.
\begin{corollary}\label{coroll:chaves}
Let $p_{XY|AB}$ be a distribution in the $(2,2,3,3)$ Bell scenario that violates at least one CHSH-type inequality. Then there exists an LOSR transformation $\mathcal{T}$, such that $\mathcal{T}(p_{XY|AB})$ violates one of the BC entropic inequalities~\eqref{eq: BCineqs}.
\end{corollary}
\begin{proof}
For each CHSH-type inequality in the $(2,2,3,3)$ scenario, there exists a coarse-graining in which two of the outcomes are mapped to one (for each party and each input) such that for any initial distribution in the $(2,2,3,3)$ scenario that violates the CHSH-type inequality the coarse-grained distribution violates one of the CHSH-inequalities in the $(2,2,2,2)$ scenario. Hence, for the given $p_{XY|AB}$, after applying the corresponding coarse-graining for the violated CHSH-type inequality, followed by the LOSR operation from Theorem~\ref{theorem: chaves} we violate one of the BC entropic inequalities.
\end{proof}
This corollary means that we can limit our analysis to $\Pi^{(2,2,3,3)}_{\chsh}$, the polytope in which all the CHSH inequalities are satisfied, and, in particular, the non-classical region of this. This is the region in which one of the $I_{2233}$ inequalities is violated.

In going from the $(2,2,2,2)$ to $(2,2,3,3)$ scenario, a new class of inequalities (the $I_{2233}$ inequalities) become relevant in probability space but the entropic characterisation remains unchanged, since entropic inequalities do not depend on the number of measurement outcomes.  It is natural to ask whether all non-classical distributions in the $(2,2,3,3)$ scenario that satisfy all the CHSH inequalities cannot be certified entropically. However, this is not the case as shown by the following proposition.

\begin{proposition}
\label{proposition: countereg}
The polytope $\Pi_{\chsh}^{(2,2,3,3)}$ is not entropically classical.
\end{proposition}
\begin{proof}
  Consider the distribution
\begin{equation}
\label{eq: countereg}
p_e:=\frac{1}{50}\, 
    \begin{array}{ |ccc|ccc|} 
 \hline
 21 & 0 & 0 & 21 & 0 & 0\\ 
 0 & 2 & 0 & 1 & 1 & 0\\
11 & 0 & \vphantom{\frac{1}{f}}16 & 0 & 1 & \vphantom{\frac{1}{f}}26\\
 \hline
 31 & 0 & 0 & 20 & 1 & 10\\ 
 1 & 1 & 0 & 1 & 0 & 1\\
 0 & 1 & \vphantom{\frac{1}{f}}16 & \vphantom{\frac{1}{f}}1 & 1 & 15\\
 \hline
\end{array}\,.
\end{equation}
This is formed by mixing the non-local extremal point number~8 of $\Pi_{\chsh}^{(2,2,3,3)}$ (see Table~\ref{tab:vertices}) with the three local deterministic points 18, 26 and 47 with respective weights $1/10$, $3/10$, $1/5$ and $2/5$, and hence is in $\Pi_{\chsh}^{(2,2,3,3)}$.  It achieves a $I_{\BC}^4$ value of $0.0199733$, in violation of $I_{\BC}^4\leq0$, so is not entropically classical.
\end{proof}

We remark that by mixing with more local deterministic distributions and varying the weights, larger violations of $I_{\BC}^4\leq0$ can be found; the distribution $p_e$ used in the previous proposition was chosen for its relative simplicity.

Interestingly, we find that the Shannon entropic BC inequalities appear 
to give the largest violation among the Tsallis entropic inequalities 
for $q\geq 1$ when applied to $p_e$. This can be seen in Figure~\ref{fig: countereg_dist}.

In light of Proposition~\ref{proposition: countereg}, it is natural to ask whether the non-classicality of all distributions in $\Pi_{\chsh}^{(2,2,3,3)}$ can be detected through entropic inequalities. We find numerical evidence that suggests the contrary, i.e., that there are non-classical distributions in $\Pi_{\chsh}^{(2,2,3,3)}$ whose non-classicality cannot be detected through entropic inequalities using a general class of post-processing operations, and hence that these entropic inequalities are not sufficient for detecting non-classicality in the $(2,2,3,3)$ scenario. Before presenting these results, we briefly overview the post-processing operations considered in this work.

\subsection{Post-processing operations}\label{sec:pp}
In this paper we study whether entropic inequalities can always detect non-classicality in the $(2,2,3,3)$ Bell scenario. In order to do so we could in principle consider applying any NCNG operation to the distribution prior to evaluating the entropic inequality. However, due to the difficulty in dealing with arbitrary NCNG operations, we consider the subset of these corresponding to LOSR+E instead.  In~\cite{Wolfe2019} it was shown that all LOSR operations can be generated by convex combinations of local deterministic operations.  These can be thought of in the following way.  Each party first does a deterministic function on their input, uses the result as the input to their device, then does a deterministic function on their input and the output of their device to form the final output.  All such operations correspond to local relabellings and local coarse-grainings. Note that deterministic classical distributions can be formed as a special case of coarse-graining (a local deterministic distribution is formed when each party coarse-grains all of their outputs to one output for each of their inputs). For the distributions we consider for our main conjectures, it turns out that all the coarse-grainings give rise to local distributions (cf.\ Proposition~\ref{proposition: CG}), so, by considering mixing with deterministic classical distributions, local relabelling and exchange of parties we can cover all LOSR+E operations. We hence start by separately considering mixing with classical distributions, and then consider relabelling and exchange of parties.

\subsection{Mixing with classical distributions}
\label{ssec: locmix}

Analogously to the $(2,2,2,2)$ case, we can define a family of distributions $p_{\iso, \epsilon}^{(2,2,3,3)}=\epsilon p_{\NL}+(1-\epsilon)p_{\noise}^{(2,2,3,3)}$, where $p_{\noise}^{(2,2,3,3)}$ is the uniform distribution with all entries equal to $1/9$ and $\epsilon\in[0,1]$. This class of distributions is isotropic in the sense that the marginal distributions are uniform for each input of each party. In order to show the insufficiency of entropic inequalities, one needs to identify at least one non-classical distribution whose non-classicality cannot be detected through entropic inequalities. We will discuss this for the class $p_{\iso, \epsilon}^{(2,2,3,3)}$ and only consider distributions of this form in the rest of the paper. Further, without loss of generality, we consider only the BC inequality $I_{\BC}^4\leq 0$ in what follows (by symmetry all the arguments will also hold for isotropic distributions corresponding to relabelled versions of $p_{\NL}$ and the corresponding BC inequalities).

\subsubsection{Using Shannon entropy}
\label{sssec: 2233shan}
In the entropic picture of the $(2,2,3,3)$ scenario, the 4 BC Inequalities~\eqref{eq: BCineqs} still hold (these are valid independently of the cardinality of the random variables). Again, analogously to the $(2,2,2,2)$ case, the maximally non-local distribution, $p_{\NL}$ has the same entropy vector as the classical distribution
\begin{equation}
    \label{eq: Pc2233}
    p_{\Cl}^{(2,2,3,3)}=
    \begin{array}{ |c|c|} 
 \hline
 \frac{1}{3} \quad 0 \quad 0 & \frac{1}{3} \quad 0 \quad 0\\ 
 0 \quad \frac{1}{3} \quad 0 & 0 \quad \frac{1}{3} \quad 0\\
 0 \quad 0 \quad \vphantom{\frac{1}{f}}\frac{1}{3} & 0 \quad 0 \quad \vphantom{\frac{1}{f}}\frac{1}{3}\\
 \hline
 \frac{1}{3} \quad 0 \quad 0 & \frac{1}{3} \quad 0 \quad 0\\ 
 0 \quad \frac{1}{3} \quad 0 & 0 \quad \frac{1}{3} \quad 0\\
 0 \quad 0 \quad \vphantom{\frac{1}{f}}\frac{1}{3} & 0 \quad 0 \quad \vphantom{\frac{1}{f}}\frac{1}{3}\\
 \hline
\end{array}
\end{equation}
(amongst others).
The distribution $p_{\NL}$ is hence entropically classical. However, in contrast to the $(2,2,2,2)$ case, we have evidence suggesting that there are values of $\epsilon$ for which $p_{\iso, \epsilon}^{(2,2,3,3)}$ is non-classical, but such that the mixture $vp_{\iso, \epsilon}^{(2,2,3,3)}+(1-v)p_{\Loc}$ is entropically classical for all classical distributions $p_{\Loc}$ and all $v\in [0,1]$, i.e., there exist non-classical distributions in the $(2,2,3,3)$ scenario for which mixing with classical distributions never gives rise to a non-classical entropy vector.

We begin by considering mixing $p_{\iso, \epsilon}^{(2,2,3,3)}$ with $p_{\Cl}^{(2,2,3,3)}$ in analogy with the treatment of the $(2,2,2,2)$ case. Although we have not fully proven this, from our numerics, this mixing appears to be optimal in the sense that when it does not allow for entropic violations, no other mixing can either. This allows us to identify a range of $\epsilon$ for which the mixture $p_{\iso, \epsilon}^{(2,2,3,3)}$ is non-classical, yet appears to remain entropically classical even when mixed with arbitrary classical distributions.  We begin with two propositions whose proofs can be found in Appendix~\ref{appendix: proofs}.

\begin{restatable}{proposition}{Propnonloc} \label{prop:nonloc} $p_{\iso, \epsilon}^{(2,2,3,3)}$ is non-classical if and only if $\epsilon>1/2$. Further, for $\epsilon\leq4/7$, $p_{\iso, \epsilon}^{(2,2,3,3)}$ satisfies all the CHSH-type inequalities, while for $\epsilon>4/7$ it violates at least one CHSH-type inequality.
\end{restatable}
By analogy with the $(2,2,2,2)$ case, we consider the violation of
$I_{\BC}^4$ attainable by mixing $p_{\iso, \epsilon}^{(2,2,3,3)}$ with
$p_{\Cl}^{(2,2,3,3)}$.  We find that for $\epsilon \in (1/2,4/7]$,
$p_{\iso, \epsilon}^{(2,2,3,3)}$ is non-classical but does not violate any
of the BC inequalities. As shown in the above proposition, these distributions are in the CHSH-classical polytope $\Pi_{\chsh}^{(2,2,3,3)}$ and hence lie in our region of interest.

\begin{restatable}{proposition}{PropBCshan}
\label{proposition: BCshan2233}
For $\epsilon \leq 4/7$,  $p_{\E,\epsilon,v}^{(2,2,3,3)}=vp_{\iso,
  \epsilon}^{(2,2,3,3)}+(1-v)p_{\Cl}^{(2,2,3,3)}$ does not violate any
of the BC entropic inequalities~\eqref{eq: BCineqs} for any $v \in [0,1]$. However, for all $\epsilon > 4/7$, there exists a $v\in [0,1]$ such that the entropic inequality $I_{\BC}^4\leq 0$ is violated by $p_{\E,\epsilon,v}^{(2,2,3,3)}$.  
\end{restatable}
The second part of this proposition already follows from Corollary~\ref{coroll:chaves} and Proposition~\ref{prop:nonloc}.

\begin{corollary}\label{cor:entcl}
For $\epsilon \leq 4/7$,  $p_{\E,\epsilon,v}^{(2,2,3,3)}=vp_{\iso,
  \epsilon}^{(2,2,3,3)}+(1-v)p_{\Cl}^{(2,2,3,3)}$ is entropically
classical for all $v\in[0,1]$.
\end{corollary}
\begin{proof}
This follows from Proposition~\ref{proposition: BCshan2233}
and Lemma~\ref{lem:BCcomplete}.
\end{proof}

While Proposition~\ref{proposition: BCshan2233} shows that the proof strategy of~\cite{Chaves13} does not directly generalize to all non-classical distributions in the $(2,2,3,3)$ case, it does not rule out the possibility that there may exist other mixings with classical distributions that could transform $p_{\iso, \epsilon}^{(2,2,3,3)}$ for $\epsilon \in (1/2,4/7]$ into a distribution that violates one of the BC inequalities. To investigate this, we can consider the polytope formed by mixing $p_{\iso, \epsilon}^{(2,2,3,3)}$ with classical distributions for some $\epsilon \leq 4/7$, i.e., the polytope $\Conv(\{p_{\iso, \epsilon}^{(2,2,3,3)}\}\bigcup\{p^{\Loc,k}\}_k)$, where $\{p^{\Loc,k}\}_k$ denotes the set of all (81) local deterministic vertices of the $(2,2,3,3)$ Bell-local polytope, $\{p_{\mathrm{iso},\epsilon}^{(2,2,3,3)}\}$ is a set with a single element and $\Conv()$ denotes the convex hull. We considered several values of $\epsilon \leq 4/7$ and numerically optimised the entropic expression $I_{\BC}^4$ over the non-classical region of this polytope\footnote{To restrict to the non-classical region, it is sufficient to mix with a subset of these 81 locals---see Appendix~\ref{appendix: evidence} for more detail.} for each\footnote{Note however that it is enough that these results hold for one value of $\epsilon \in (1/2,4/7]$ in order to conclude that entropic inequalities are not sufficient for detecting non-classicality in this scenario.} and were unable to find violations. The optimization involves a non-linear objective function with linear constraints.  Hence, it is possible that the numerical approach missed the global optimum. Nevertheless, this is evidence for the following conjecture and is presented in more detail in Appendix~\ref{appendix: evidence}.  Proposition~\ref{proposition: BCshan2233} along with the figures and evidence in Appendices~\ref{appendix: figures} and~\ref{appendix: evidence} also suggest this conjecture.

\begin{conjecture}
\label{conjecture: 2233shan}
Let $\epsilon\leq4/7$. For all mixtures of the distribution $p_{\iso,\epsilon}^{(2,2,3,3)}$ with classical distributions in the $(2,2,3,3)$ Bell scenario, the resulting distribution is entropically classical, i.e., all distributions in $\Conv(\{p_{\iso, \epsilon}^{(2,2,3,3)}\}\bigcup\{p^{\Loc,k}\}_k\})$ are entropically classical.
\end{conjecture}

The interesting cases of Conjecture~\ref{conjecture: 2233shan} are for non-classical distributions (i.e., for $\epsilon>1/2$), and the most relevant of these are those that can be achieved in quantum theory. The next remark addresses this case.

\begin{remark}
  There exist non-classical quantum distributions that lie in the polytope $\Conv(\{p_{\iso, \epsilon=4/7}^{(2,2,3,3)}\}\bigcup\{p^{\Loc,k}\}_k\})$ and in this case, our results suggest that the non-classicality of the corresponding distributions cannot be detected through entropic inequalities as we now explain. Let $p_{\QM}$ be the quantum distribution from~\cite[Equation~(14) with $d=3$]{CGLMP02} with Bob's inputs relabelled. This violates $I_{\BC}^4\leq 0$ through mixing with $p_{\Cl}^{(2,2,3,3)}$. Consider then mixing $p_{\QM}$ with uniform noise $p_{\noise}^{(2,2,3,3)}$ to obtain $p_{\mix}(u):=up_{\QM}+(1-u)p_{\noise}^{(2,2,3,3)}$ ($u\in [0,1]$). We found that for some values of $u$ (e.g., $u=7/10$), $p_{\mix}(u)$ is non-classical. 
  Further, $p_{\mix}(u)$ is quantum achievable since it can be obtained from the density operator $u\ket{\psi'}\bra{\psi'}+(1-u)\frac{\mathbb{I}}{9}$ (where $\ket{\psi'}$ is the two qutrit state producing $p_{\QM}$) and the same quantum measurements that produce $p_{\QM}$ from $\ket{\psi'}$.
\end{remark}

\subsubsection{Using Tsallis entropies}
\label{sssec: 2233tsal}

Given the results (Proposition~\ref{proposition: BCshan2233} and Conjecture~\ref{conjecture: 2233shan}) of the previous section for Shannon entropic inequalities, a natural question is whether other entropic measures can provide an advantage over the Shannon entropy in detecting non-classicality.  Here, we look at Tsallis entropies and find that similar results hold in this case as well, suggesting that Tsallis entropies also do not allow us to completely solve the problem.

The properties of monotonicity, strong-subadditivity and the chain rule are sufficient to derive the BC inequalities, which hence also hold for Tsallis entropy when $q\geq 1$. (Other generalized entropies such as R\'enyi or min/max entropies do not satisfy one or more of these properties in general and it is not clear whether the analogues of~\eqref{eq: BCineqs} hold for these.) In other words, for all $q\geq 1$ we have
\begin{widetext}
\begin{equation}
\label{eq: BCineqsTs}
    \begin{split}
I_{\BC,q}^1=S_q(X_0Y_0)+S_q(X_1)+S_q(Y_1)-S_q(X_0Y_1)-S_q(X_1Y_0)-S_q(X_1Y_1)\leq 0\\
        I_{\BC,q}^2=S_q(X_0Y_1)+S_q(X_1)+S_q(Y_0)-S_q(X_0Y_0)-S_q(X_1Y_0)-S_q(X_1Y_1)\leq 0\\
        I_{\BC,q}^3=S_q(X_1Y_0)+S_q(X_0)+S_q(Y_1)-S_q(X_0Y_0)-S_q(X_0Y_1)-S_q(X_1Y_1)\leq 0\\
        I_{\BC,q}^4=S_q(X_1Y_1)+S_q(X_0)+S_q(Y_0)-S_q(X_0Y_0)-S_q(X_0Y_1)-S_q(X_1Y_0)\leq 0\\
    \end{split}
\end{equation}
\end{widetext}
and we refer to these as the Tsallis entropic BC inequalities. Entropic classicality in Tsallis entropy space can be defined analogously to Definition~\ref{definition: entclass}, in terms of Tsallis entropy vectors over the set of variables $\mathcal{S}$ (Equation~\eqref{eq: coexisting}). We say that a distribution is \emph{$q$-entropically classical} if its entropy vector written in terms of the Tsallis entropy of order $q$ is achievable using a classical distribution.  In the case of the Shannon entropy, we used the fact (Lemma~\ref{lem:BCcomplete}) that the BC Inequalities~\eqref{eq: BCineqs} are known to be necessary and sufficient for entropic classicality for 2-input Bell scenarios \cite{Fritz13}. However, it is not clear if the result of \cite{Fritz13} generalises to Tsallis entropies for $q>1$. Thus our results in the Tsallis case are weaker than those for Shannon, being stated only for the BC inequalities.  We leave the generalization to arbitrary Tsallis entropic inequalities as an open problem.

\begin{restatable}{proposition}{PropBCtsal}
\label{proposition: BCtsal2233}
For $\epsilon \leq 4/7$,  $p_{\E,\epsilon,v}^{(2,2,3,3)}=vp_{\iso,
  \epsilon}^{(2,2,3,3)}+(1-v)p_{\Cl}^{(2,2,3,3)}$ does not violate
any of the Tsallis BC inequalities~\eqref{eq: BCineqsTs} for any $v \in
[0,1]$ and $q>1$. However, for $\epsilon > 4/7$ and every $q>1$, there always exists a $v\in [0,1]$ such that the entropic inequality $I_{\BC,q}^4\leq 0$ is violated by $p_{\E,\epsilon,v}^{(2,2,3,3)}$.  
\end{restatable}

We refer the reader to Appendix~\ref{appendix: proofs} for a proof of
this Proposition. To investigate the extension to other mixings, we tried the same computational procedure (see Appendix~\ref{appendix: evidence}) as in the Shannon case. We found no violation of the Tsallis entropic BC inequalities for any mixings of $p_{\iso, \epsilon}^{(2,2,3,3)}$ with classical distributions, for several values of $q>1$ and $\epsilon \in (1/2,4/7]$, leading to the following conjecture, which is similar to Conjecture~\ref{conjecture: 2233shan}.
\begin{conjecture}
\label{conjecture: 2233tsal}
Let $\epsilon\leq4/7$. For all mixtures of the distribution $p_{\iso, \epsilon}^{(2,2,3,3)}$ with classical distributions in the $(2,2,3,3)$ Bell scenario, the resulting distribution does not violate any of the Tsallis entropic BC inequalities for any $q>1$, i.e., all distributions in $\Conv(\{p_{\iso, \epsilon}^{(2,2,3,3)}\}\bigcup\{p^{\Loc,k}\}_k)$ satisfy the Tsallis entropic BC Inequalities~\eqref{eq: BCineqsTs} for all $q>1$.
\end{conjecture}

Figure~\ref{fig: 2233shantsal1}, shows the values of $\epsilon$ and $v$ for which $I_{\BC,q}^4$ (for $q=1, 2 , 8$) evaluated with $p_{\E,\epsilon,v}^{(2,2,3,3)}$ is positive, which is also suggestive of this conjecture.

\begin{remark}
  Any impossibility result for the $(2,2,3,3)$ scenario also holds in the $(2,2,d,d)$ case for $d>3$ because the former is always embedded in the latter i.e., every distribution in the $(2,2,3,3)$ scenario has a corresponding distribution in all the $(2,2,d,d)$ scenarios with $d>3$ which can be obtained by assigning a zero probability to the additional outcomes. Further, the entropic Inequalities~\eqref{eq: BCineqs} remain the same for all these scenarios as they do not depend on the cardinality of the random variables involved. Thus the existence non-classical distributions for the $d=3$ case whose non-classicality cannot be detected by entropic inequalities implies the same result for all $d>3$.
\end{remark}

\subsection{Beyond classical mixings}
\label{ssec: CGrelab}

So far, we only considered mixing with classical distributions to obtain entropic violations and gave evidence that this does not work for some non-classical distributions in the $(2,2,3,3)$ scenario. This motivates us to study whether using arbitrary LOSR+E operations allows us to detect this non-classicality through entropic violations. We show in this section that if Conjectures~\ref{conjecture: 2233shan} and~\ref{conjecture: 2233tsal} hold then they also hold for all LOSR+E operations. First consider the following example.

The maximum possible violation of the BC inequalities in the $(2,2,2,2)$ case is $I_{\BC}^4=\ln 2$~\cite{Chaves13}. This is derived by considering only Shannon inequalities within the coexisting sets, and the bound that the maximum entropy of a binary variable is $\ln 2$.  An analogous proof holds in the $(2,2,3,3)$ case, except that the bound is then $\ln 3$. In the former case we have $p_{\E,\epsilon=1,v=1/2}=\frac{1}{2}p_{\PR}+\frac{1}{2}p_{\Cl}$, which maximally violates $I_{\BC}^4\leq0$, while in the latter case, one such distribution is formed by $(p_{\NL}+p_{\NL}^*+p_{\Cl}^{(2,2,3,3)})/3$, where $p_{\NL}^*$ is another extremal non-local distribution:
\begin{equation}
    \label{eq: NL*}
    p_{\NL}^*=
    \begin{array}{ |c|c|} 
 \hline
 \frac{1}{3} \quad 0 \quad 0 & \frac{1}{3} \quad 0 \quad 0\\ 
 0 \quad \frac{1}{3} \quad 0 & 0 \quad \frac{1}{3} \quad 0\\
 0 \quad 0 \quad \vphantom{\frac{1}{f}}\frac{1}{3} & 0 \quad 0 \quad \vphantom{\frac{1}{f}}\frac{1}{3}\\
 \hline
 \frac{1}{3} \quad 0 \quad 0 & 0 \quad 0 \quad \frac{1}{3}\\ 
 0 \quad \frac{1}{3} \quad 0 & \frac{1}{3} \quad 0 \quad 0\\
 0 \quad 0 \quad \vphantom{\frac{1}{f}}\frac{1}{3} & 0 \quad \vphantom{\frac{1}{f}}\frac{1}{3} \quad 0\\
 \hline
\end{array}\,.
\end{equation}

Since the equal mixture $(p_{\NL}+p_{\Cl}^{(2,2,3,3)})/2$ violates
$I_{\BC}^4\leq0$ non-maximally, one may be motivated to use the
non-local distribution $\tilde{p}_{\NL}=(p_{\NL}+p_{\NL}^*)/2$ in place
of $p_{\NL}$ in the
definition of $p_{\iso,\epsilon}^{(2,2,3,3)}$, i.e., to take
\begin{equation*}
    \tilde{p}_{\iso,\epsilon}^{(2,2,3,3)}=\epsilon \tilde{p}_{\NL}+(1-\epsilon)p_{\Cl}^{(2,2,3,3)}.
\end{equation*}
One could then consider whether for $\epsilon \in (1/2,4/7]$, $\tilde{p}_{\E,\epsilon,v}^{(2,2,3,3)}=v\tilde{p}_{\iso,\epsilon}^{(2,2,3,3)}+(1-v)p_{\Cl}^{(2,2,3,3)}$ violates $I_{\BC}^4\leq0$. Interestingly, while $\tilde{p}_{\E,\epsilon,v}^{(2,2,3,3)}$ violates $I_{\BC}^4\leq0$ for a larger range of $v$ values whenever $\epsilon> 4/7$, it does not give any violation (for any value of $v$) when $\epsilon \leq 4/7$, and Propositions~\ref{prop:nonloc} and \ref{proposition: BCshan2233} also hold if $\tilde{p}_{\iso,\epsilon}^{(2,2,3,3)}$ replaces $p_{\iso,\epsilon}^{(2,2,3,3)}$ (see Figure~\ref{fig: 2233shantsal2} for an illustration). The corresponding results also hold for the Tsallis case with $q>1$, i.e., Proposition~\ref{proposition: BCtsal2233} also holds with $\tilde{p}_{\iso,\epsilon}^{(2,2,3,3)}$ replacing $p_{\iso,\epsilon}^{(2,2,3,3)}$ (see also Figure~\ref{fig: 2233shantsal}). These suggest that mixing with relabellings in addition to mixing with classical distributions may also not help to violate entropic inequalities when $\epsilon \leq 4/7$.\bigskip

In the remainder of this section we consider the full set of LOSR+E operations. We first note that all input coarse-grainings of $p_{\iso,\epsilon}^{(2,2,3,3)}$ result in local distributions (there are no Bell inequalities if one party has only one input). Similarly, considering output coarse-grainings, whenever three outcomes are mapped to one the resulting distribution is always classical because there are no Bell inequalities if one party always makes a fixed outcome for one of their inputs. We henceforth only consider coarse-grainings that take two outcomes to one. We can choose two of the three outcomes to combine into one for each party and each local input. For the four input choices $\{A=0,A=1,B=0,B=1\}$, there are 81, 108, 54 and 12 distinct coarse-grainings of this type when the outcomes of either 4, 3, 2 or 1 input choices are coarse-grained. Thus there are a total of $255$ coarse-grainings that remain.

If we apply all such coarse-grainings to $p_{\iso, \epsilon}^{(2,2,3,3)}$, this generates $255$ possible distributions that we denote $\{p^{\CG,i}_\epsilon\}_i$, $i\in[255]$. There are also $432$ distinct local relabellings of $p_{\iso,\epsilon}^{(2,2,3,3)}$, which we denote by $\{p^{\RL,j}_\epsilon\}_j$, $j\in[432]$ (this set includes $p_{\iso,\epsilon}^{(2,2,3,3)}=p^{\RL,1}_\epsilon$). Due to symmetries of $p_{\iso,\epsilon}^{(2,2,3,3)}$, it turns out that exchanging parties can be achieved through local relabellings for these distributions, so we do not need to separately consider the exchange in our results pertaining to $p_{\iso,\epsilon}^{(2,2,3,3)}$.  The set of all distributions that can be achieved through a convex mixture of $p_{\iso,\epsilon}^{(2,2,3,3)}$ with its coarse-grainings, relabellings and classical distributions is a convex polytope $\Pi_{\epsilon}$ for each $\epsilon$ and is the convex hull of these $255+432+81=768$ points, i.e.,
$$\Pi_{\epsilon}:=\Conv\left(\{p^{\CG,i}_\epsilon\}_i\bigcup\{p^{\RL,j}_\epsilon\}_j\bigcup\{p^{\Loc,k}\}_k\}\right)\,.$$
We present the results for the remaining coarse-grainings and relabellings separately below. Firstly, we show that the coarse-grainings $\{p^{\CG,i}_\epsilon\}_i$ of $p_{\iso, \epsilon}^{(2,2,3,3)}$ are classical if and only if $\epsilon\leq 4/7$.

\begin{restatable}{proposition}{PropCGs}
\label{proposition: CG}
The distribution $p^{\CG,i}_\epsilon$ is classical for all $i$ if and
only if $\epsilon \leq 4/7$.
\end{restatable}
This is intuitive because $p_{\iso,\epsilon}^{(2,2,3,3)}$ satisfies all the CHSH-type inequalities if and only if $\epsilon\leq4/7$. Since coarse-grainings cannot generate non-classicality and correspond to reducing the number of outcomes, and since $I_{2233}$ requires three outcomes, after coarse-graining the only relevant thing is whether there is a CHSH-violation.  A full proof is given in Appendix~\ref{appendix: proofs}.

Proposition~\ref{proposition: CG} implies that $\Pi_{\epsilon}=\Conv(\{p^{\RL,j}_\epsilon\}_j\bigcup\{p^{\Loc,k}\}_k)$ $\forall \epsilon \leq 4/7$, and that it is not necessary to consider coarse-grainings for such values of $\epsilon$. Our next results are that if Conjectures~\ref{conjecture: 2233shan} and~\ref{conjecture: 2233tsal} hold for $p_{\iso, \epsilon}^{(2,2,3,3)}$ for $\epsilon \leq 4/7$, then they continue to hold even when we consider arbitrary convex combinations with classical distributions and local relabellings of $p_{\iso, \epsilon}^{(2,2,3,3)}$.

\begin{restatable}{proposition}{ProprelabsA}
\label{proposition: relab1}
Let $\epsilon\leq4/7$. If Conjecture~\ref{conjecture: 2233shan} holds, then every distribution in $\Pi_{\epsilon}$ is Shannon entropically classical. 
\end{restatable}

\begin{restatable}{proposition}{ProprelabsB}
\label{proposition: relab2}
Let $\epsilon\leq4/7$. If Conjecture~\ref{conjecture: 2233tsal} holds, then every distribution in $\Pi_{\epsilon}$ satisfies the Tsallis entropic BC Inequalities~\eqref{eq: BCineqsTs} $\forall q>1$. 
\end{restatable}
These are proven in Appendix~\ref{appendix: proofs} and give the
following corollary.

\begin{corollary}\label{corollary: final}
  Let $\epsilon\leq4/7$. If Conjectures~\ref{conjecture: 2233shan} and~\ref{conjecture: 2233tsal} hold, then for any operation $\mathcal{O}$ in LOSR+E, $\mathcal{O}(p_{\iso, \epsilon}^{(2,2,3,3)})$ does not violate a Shannon or Tsallis ($q>1$) entropic BC inequality.
 \end{corollary}
  
\section{Discussion}
\label{sec: disc&conc}
We have provided evidence that there are distributions in the $(2,2,3,3)$ scenario for which arbitrary LOSR+E operations do not enable detection of non-classicality with any Shannon entropic inequalities or Tsallis entropic BC inequalities. This is in contrast to the $(2,2,2,2)$ scenario~\cite{Chaves13}, where for any non-classical distribution, there is always a simple LOSR operation that results in a distribution violating one of the Shannon BC inequalities. In order that BC inequalities do not detect non-classicality we need that the distributions are non-classical while at the same time satisfying all CHSH-type Bell inequalities. Having found all the vertices that characterize this region, we identify distributions in this region that violate BC entropic inequalities. Thus, the set of all non-classical distributions in the $(2,2,3,3)$ scenario that cannot be certified through entropic inequalities under LOSR+E post-processings is not characterized by the CHSH-type inequalities.

Although we considered LOSR+E operations, a natural next question is to what extent the results can be generalized to more general NCNG operations. In particular, there could be a non-linear NCNG map that allows the entropic BC inequalities to detect a wider range of non-classical distributions. It would be interesting to see whether for any non-classical distribution of the form $p_{\iso,\epsilon}^{(2,2,3,3)}$ with $1/2<\epsilon\leq 4/7$ (conjectured to be entropically classical with respect to LOSR+E), one of these more general operations would allow its non-classicality to be detected entropically. We leave this as an open question. 

In~\cite{Wajs15} it was shown that (without post-processing) Tsallis entropic inequalities can detect non-classicality undetectable by Shannon entropic inequalities in the $(2,2,2,2)$ and $(2,2,3,3)$ Bell scenarios. In the presence of LOSR+E operations, we did not find any advantage of Tsallis entropies over the Shannon entropy in the $(2,2,3,3)$ Bell scenario. In fact, for some non-classical distributions such as that of Equation~\eqref{eq: countereg}, the Shannon entropic inequalities appear to give the largest violations among Tsallis entropies with $q\geq1$ (which corresponds to those for which the BC inequalities can be derived in the classical case). On the other hand, for the family of distributions $p_{\iso,\epsilon}^{(2,2,3,3)}$, our results suggest that the range of $\epsilon$ for which post-processing via mixing with classical distributions enables non-classicality detection is the same for the Shannon as well as Tsallis entropic BC inequalities for $q>1$. However, when entropic detection of non-classicality is possible, using Tsallis entropy can make it easier to do this detection in the sense that there is a wider range of mixings that achieve this (see Figure~\ref{fig: 2233shantsal}).\footnote{As a specific example, consider the distribution $p_{\E,\epsilon=0.7,v=0.4}^{(2,2,3,3)}$. Figure~\ref{fig: 2233shantsal1} indicates that this distribution violates the Tsallis entropic inequality $I_{\BC,q=2}^4 \leq 0$ but does not violate the Shannon entropic inequality $I_{\BC}^4 \leq 0$. However, we can always further mix $p_{\E,\epsilon=0.7,v=0.4}^{(2,2,3,3)}$ with the classical distribution $p_{\Cl}^{(2,2,3,3)}$ to obtain $0.05p_{\E,\epsilon=0.7,v=0.4}^{(2,2,3,3)}+0.95p_{\Cl}^{(2,2,3,3)}=0.02p_{\iso, \epsilon}^{(2,2,3,3)}+0.98p_{\Cl}^{(2,2,3,3)}$ which violates the Shannon entropic inequality $I_{\BC}^4\leq0$. This is also in agreement with the results of \cite{Wajs15}, since when mixing is not considered, we also find examples where it is advantageous to use Tsallis entropy in the $(2,2,3,3)$ scenario.}

In conclusion, while the entropic approach for detecting non-classicality is useful in a number of scenarios, it is known to have disadvantages in others.  In particular, in the absence of post-selection we are not aware of any cases where entropic inequalities can be violated~\cite{Weilenmann16, Weilenmann18, Vilasini2019}. Here, we find that the entropic approach also suffers drawbacks in the presence of post-selection as it may fail to detect non-classicality under a natural class of post-processing operations, both in the case of Shannon and Tsallis entropies. 
However, 
this method remains of use since in many cases non-classicality can be detected using it.

\begin{acknowledgements}
We thank Mirjam Weilenmann for useful discussions. VV acknowledges financial support from the Department of Mathematics, University of York. RC is supported by EPSRC's Quantum Communications Hub (grant numbers EP/M013472/1 and EP/T001011/1) and by an EPSRC First Grant (grant number EP/P016588/1).
\end{acknowledgements}


\onecolumngrid
\appendix

\section*{\textsc{Appendix}}
\section{Characterisation of \texorpdfstring{$\Pi_{\chsh}^{(2,2,3,3)}$}{the CHSH-classical region}}
\label{appendix: vertices}
Table~\ref{tab:vertices} enumerates the 47 extremal points of the polytope $\Pi_{\chsh}^{(2,2,3,3)}$ that saturate or violate the inequality ${I_{2233}\leq2}$~\eqref{eq: CGLMP} while satisfying all the CHSH-type inequalities in the $(2,2,3,3)$ Bell scenario. The first 17 of these are non-classical while the remaining 30 are local deterministic vertices. Due to Proposition~\ref{proposition: probspace} and the symmetries of the scenario, the remaining vertices of $\Pi_{\chsh}^{(2,2,3,3)}$ can be generated by taking the orbit of these vertices under local relabellings and exchange of parties. In Table~\ref{tab:vertices}, each extremal point is given by a single 36 dimensional vector which corresponds to writing the point in the notation explained in Section~\ref{ssec: prob_ent} (a $6\times 6$ matrix) and ``flattening'' it by writing one row after another in order.
{\setlength{\tabcolsep}{1em}
\begin{table}[]
    \centering
    \begin{tabular}{|c|c|}
    \hline
    Number & Vertex\\
    \hline
       1  &  $\frac{1}{6}$(1, 1, 0, 1, 0, 1, 0, 1, 1, 1, 1, 0, 1, 0, 1, 0, 1, 1, 1, 0, 1, 0, 2, 0, 1, 1, 0, 0, 0, 2, 0, 1, 1, 2, 0, 0) \\
       2  &  $\frac{1}{6}$(1, 1, 0, 1, 0, 1, 0, 1, 1, 1, 1, 0,  1, 0, 1, 0, 1, 1, 2, 0, 0, 0, 1, 1, 0, 2, 0, 1, 0, 1, 0, 0, 2, 1, 1, 0)\\
       3  &  $\frac{1}{6}$(1, 1, 0, 2, 0, 0, 0, 1, 1, 0, 2, 0, 1, 0, 1, 0, 0, 2, 1, 0, 1, 0, 1, 1, 1, 1, 0, 1, 0, 1, 0, 1, 1, 1, 1, 0)\\
       4 &  $\frac{1}{6}$(2, 0, 0, 1, 0, 1, 0, 2, 0, 1, 1, 0, 0, 0, 2, 0, 1, 1, 1, 0, 1, 0, 1, 1, 1, 1, 0, 1, 0, 1, 0, 1, 1, 1, 1, 0)\\
       5  &  $\frac{1}{5}$(1, 0, 0, 1, 0, 0, 0, 1, 1, 1, 1, 0, 1, 0, 1, 0, 0, 2, 1, 0, 1, 0, 1, 1, 1, 1, 0, 1, 0, 1, 0, 0, 1, 1, 0, 0)\\
       6  &  $\frac{1}{5}$(1, 0, 0, 1, 0, 0, 0, 1, 1, 1, 1, 0, 1, 0, 1, 0, 1, 1, 1, 0, 0, 0, 1, 0, 1, 1, 0, 1, 0, 1, 0, 0, 2, 1, 1, 0)\\
       7  &  $\frac{1}{5}$(1, 0, 0, 1, 0, 0, 0, 1, 1, 1, 1, 0, 1, 0, 1, 0, 1, 1, 1, 0, 1, 0, 2, 0, 1, 1, 0, 1, 0, 1, 0, 0, 1, 1, 0, 0)\\
       8  &  $\frac{1}{5}$(1, 0, 0, 1, 0, 0, 0, 2, 0, 1, 1, 0, 1, 0, 1, 0, 1, 1, 1, 0, 0, 0, 1, 0, 1, 1, 0, 1, 0, 1, 0, 1, 1, 1, 1, 0)\\
       9  &  $\frac{1}{5}$(1, 1, 0, 1, 0, 1, 0, 1, 0, 0, 1, 0, 0, 0, 2, 0, 1, 1, 1, 0, 1, 0, 1, 1, 0, 1, 0, 0, 0, 1, 0, 1, 1, 1, 1, 0)\\
       10  &  $\frac{1}{5}$(1, 1, 0, 1, 0, 1, 0, 1, 0, 0, 1, 0, 1, 0, 1, 0, 1, 1, 1, 0, 0, 0, 1, 0, 1, 1, 0, 0, 0, 2, 0, 1, 1, 1, 1, 0)\\
       11 &  $\frac{1}{5}$(1, 1, 0, 1, 0, 1, 0, 1, 0, 0, 1, 0, 1, 0, 1, 0, 1, 1, 2, 0, 0, 0, 1, 1, 0, 1, 0, 0, 0, 1, 0, 1, 1, 1, 1, 0)\\
       12 &  $\frac{1}{5}$(1, 1, 0, 1, 0, 1, 0, 1, 1, 0, 2, 0, 0, 0, 1, 0, 0, 1, 1, 0, 1, 0, 1, 1, 0, 1, 0, 0, 0, 1, 0, 1, 1, 1, 1, 0)\\
       13 &  $\frac{1}{5}$(1, 1, 0, 1, 0, 1, 0, 1, 1, 1, 1, 0, 0, 0, 1, 0, 0, 1, 1, 0, 1, 0, 1, 1, 0, 1, 0, 0, 0,1, 0, 1, 1, 2, 0, 0)\\
       14 &  $\frac{1}{5}$(1, 1, 0, 1, 0, 1, 0, 1, 1, 1, 1, 0, 0, 0, 1, 0, 0, 1, 1, 0, 1, 0, 1, 1, 0, 2, 0, 1, 0, 1, 0, 0, 1, 1, 0, 0)\\
       15 &  $\frac{1}{5}$(1, 1, 0, 2, 0, 0, 0, 1, 0, 0, 1, 0, 1, 0, 1, 0, 1, 1, 1, 0, 0, 0, 1, 0, 1, 1, 0, 1, 0, 1, 0, 1, 1, 1, 1, 0)\\
       16 &  $\frac{1}{5}$(2, 0, 0, 1, 0, 1, 0, 1, 1, 1, 1, 0, 0, 0, 1, 0, 0, 1, 1, 0, 1, 0, 1, 1, 1, 1, 0, 1, 0, 1, 0, 0, 1, 1, 0, 0)\\
       17 &  $\frac{1}{9}$(2, 1, 0, 2, 0, 1, 0, 2, 1, 1, 2, 0, 1, 0, 2, 0, 1, 2, 2, 0, 1, 0, 2, 1, 1, 2, 0, 1, 0, 2, 0, 1, 2, 2, 1, 0)\\
       18 &  (0, 0, 0, 0, 0, 0, 0, 0, 0, 0, 0, 0, 0, 0, 1, 0, 0, 1, 0, 0, 0, 0, 0, 0, 0, 0, 0, 0, 0, 0, 0, 0, 1, 0, 0, 1)\\
       19 &  (0, 0, 0, 0, 0, 0, 0, 0, 0, 0, 0, 0, 0, 0, 1, 0, 0, 1, 0, 0, 0, 0, 0, 0, 0, 0, 1, 0, 0, 1, 0, 0, 0, 0, 0, 0)\\
       20 &  (0, 0, 0, 0, 0, 0, 0, 0, 0, 0, 0, 0, 0, 0, 1, 0, 0, 1, 0, 0, 1, 0, 0, 1, 0, 0, 0, 0, 0, 0, 0, 0, 0, 0, 0, 0)\\
       21 &  (0, 0, 0, 0, 0, 0, 0, 0, 0, 0, 0, 0, 0, 0, 1, 0, 1, 0, 0, 0, 0, 0, 0, 0, 0, 0, 0, 0, 0, 0, 0, 0, 1, 0, 1, 0)\\
       22 &  (0, 0, 0, 0, 0, 0, 0, 0, 0, 0, 0, 0, 0, 0, 1, 0, 1, 0, 0, 0, 1, 0, 1, 0, 0, 0, 0, 0, 0, 0, 0, 0, 0, 0, 0, 0)\\
       23 &  (0, 0, 0, 0, 0, 0, 0, 0, 0, 0, 0, 0, 0, 0, 1, 1, 0, 0, 0, 0, 0, 0, 0, 0, 0, 0, 0, 0, 0, 0, 0, 0, 1, 1, 0, 0)\\
       24 &  (0, 0, 0, 0, 0, 0, 0, 0, 0, 0, 0, 0, 0, 1, 0, 0, 0, 1, 0, 0, 0, 0, 0, 0, 0, 1, 0, 0, 0, 1, 0, 0, 0, 0, 0, 0)\\
       25 &  (0, 0, 0, 0, 0, 0, 0, 0, 0, 0, 0, 0, 1, 0, 0, 0, 0, 1, 0, 0, 0, 0, 0, 0, 1, 0, 0, 0, 0, 1, 0, 0, 0, 0, 0, 0)\\
       26 &  (0, 0, 0, 0, 0, 0, 0, 0, 0, 0, 0, 0, 1, 0, 0, 0, 0, 1, 1, 0, 0, 0, 0, 1, 0, 0, 0, 0, 0, 0, 0, 0, 0, 0, 0, 0)\\
       27 &  (0, 0, 0, 0, 0, 0, 0, 0, 0, 0, 0, 0, 1, 0, 0, 0, 1, 0, 1, 0, 0, 0, 1, 0, 0, 0, 0, 0, 0, 0, 0, 0, 0, 0, 0, 0)\\
       28 &  (0, 0, 0, 0, 0, 0, 0, 0, 1, 0, 1, 0, 0, 0, 0, 0, 0, 0, 0, 0, 0, 0, 0, 0, 0, 0, 0, 0, 0, 0, 0, 0, 1, 0, 1, 0)\\
       29 &  (0, 0, 0, 0, 0, 0, 0, 0, 1, 0, 1, 0, 0, 0, 0, 0, 0, 0, 0, 0, 1, 0, 1, 0, 0, 0, 0, 0, 0, 0, 0, 0, 0, 0, 0, 0)\\
       30 &  (0, 0, 0, 0, 0, 0, 0, 0, 1, 1, 0, 0, 0, 0, 0, 0, 0, 0, 0, 0, 0, 0, 0, 0, 0, 0, 0, 0, 0, 0, 0, 0, 1, 1, 0, 0)\\
       31 &  (0, 0, 0, 0, 0, 0, 0, 1, 0, 0, 0, 1, 0, 0, 0, 0, 0, 0, 0, 0, 0, 0, 0, 0, 0, 1, 0, 0, 0, 1, 0, 0, 0, 0, 0, 0)\\
       32 &  (0, 0, 0, 0, 0, 0, 0, 1, 0, 0, 1, 0, 0, 0, 0, 0, 0, 0, 0, 0, 0, 0, 0, 0, 0, 0, 0, 0, 0, 0, 0, 1, 0, 0, 1, 0)\\
       33 &  (0, 0, 0, 0, 0, 0, 0, 1, 0, 0, 1, 0, 0, 0, 0, 0, 0, 0, 0, 0, 0, 0, 0, 0, 0, 1, 0, 0, 1, 0, 0, 0, 0, 0, 0, 0)\\
       34 &  (0, 0, 0, 0, 0, 0, 0, 1, 0, 0, 1, 0, 0, 0, 0, 0, 0, 0, 0, 1, 0, 0, 1, 0, 0, 0, 0, 0, 0, 0, 0, 0, 0, 0, 0, 0)\\
       35 &  (0, 0, 0, 0, 0, 0, 0, 1, 0, 1, 0, 0, 0, 0, 0, 0, 0, 0, 0, 0, 0, 0, 0, 0, 0, 0, 0, 0, 0, 0, 0, 1, 0, 1, 0, 0)\\
       36 &  (0, 0, 0, 0, 0, 0, 0, 1, 0, 1, 0, 0, 0, 0, 0, 0, 0, 0, 0, 0, 0, 0, 0, 0, 0, 1, 0, 1, 0, 0, 0, 0, 0, 0, 0, 0)\\
       37 &  (0, 0, 0, 0, 0, 0, 1, 0, 0, 0, 1, 0, 0, 0, 0, 0, 0, 0, 1, 0, 0, 0, 1, 0, 0, 0, 0, 0, 0, 0, 0, 0, 0, 0, 0, 0)\\
       38 &  (0, 0, 1, 1, 0, 0, 0, 0, 0, 0, 0, 0, 0, 0, 0, 0, 0, 0, 0, 0, 0, 0, 0, 0, 0, 0, 0, 0, 0, 0, 0, 0, 1, 1, 0, 0)\\
       39 &  (0, 1, 0, 0, 0, 1, 0, 0, 0, 0, 0, 0, 0, 0, 0, 0, 0, 0, 0, 0, 0, 0, 0, 0, 0, 1, 0, 0, 0, 1, 0, 0, 0, 0, 0, 0)\\
       40 &  (0, 1, 0, 1, 0, 0, 0, 0, 0, 0, 0, 0, 0, 0, 0, 0, 0, 0, 0, 0, 0, 0, 0, 0, 0, 0, 0, 0, 0, 0, 0, 1, 0, 1, 0, 0)\\
       41 &  (0, 1, 0, 1, 0, 0, 0, 0, 0, 0, 0, 0, 0, 0, 0, 0, 0, 0, 0, 0, 0, 0, 0, 0, 0, 1, 0, 1, 0, 0, 0, 0, 0, 0, 0, 0)\\
       42 &  (1, 0, 0, 0, 0, 1, 0, 0, 0, 0, 0, 0, 0, 0, 0, 0, 0, 0, 0, 0, 0, 0, 0, 0, 1, 0, 0, 0, 0, 1, 0, 0, 0, 0, 0, 0)\\
       43 &  (1, 0, 0, 0, 0, 1, 0, 0, 0, 0, 0, 0, 0, 0, 0, 0, 0, 0, 1, 0, 0, 0, 0, 1, 0, 0, 0, 0, 0, 0, 0, 0, 0, 0, 0, 0)\\
       44 &  (1, 0, 0, 0, 1, 0, 0, 0, 0, 0, 0, 0, 0, 0, 0, 0, 0, 0, 1, 0, 0, 0, 1, 0, 0, 0, 0, 0, 0, 0, 0, 0, 0, 0, 0, 0)\\
       45 &  (1, 0, 0, 1, 0, 0, 0, 0, 0, 0, 0, 0, 0, 0, 0, 0, 0, 0, 0, 0, 0, 0, 0, 0, 0, 0, 0, 0, 0, 0, 1, 0, 0, 1, 0, 0)\\
       46 &  (1, 0, 0, 1, 0, 0, 0, 0, 0, 0, 0, 0, 0, 0, 0, 0, 0, 0, 0, 0, 0, 0, 0, 0, 1, 0, 0, 1, 0, 0, 0, 0, 0, 0, 0, 0)\\
       47 &  (1, 0, 0, 1, 0, 0, 0, 0, 0, 0, 0, 0, 0, 0, 0, 0, 0, 0, 1, 0, 0, 1, 0, 0, 0, 0, 0, 0, 0, 0, 0, 0, 0, 0, 0, 0)\\
       \hline
    \end{tabular}
    \caption{The vertices of $\Pi_{\chsh}^{(2,2,3,3)}$ that saturate or violate the $I_{2233}$ Inequality~\eqref{eq: CGLMP}. All the vertices of the polytope can be obtained from the vertices listed here through local relabellings or exchange of parties.}
    \label{tab:vertices}
\end{table}}

\section{Plots}
\label{appendix: figures}
\begin{figure}[H]
\centering
    \includegraphics[scale=.3]{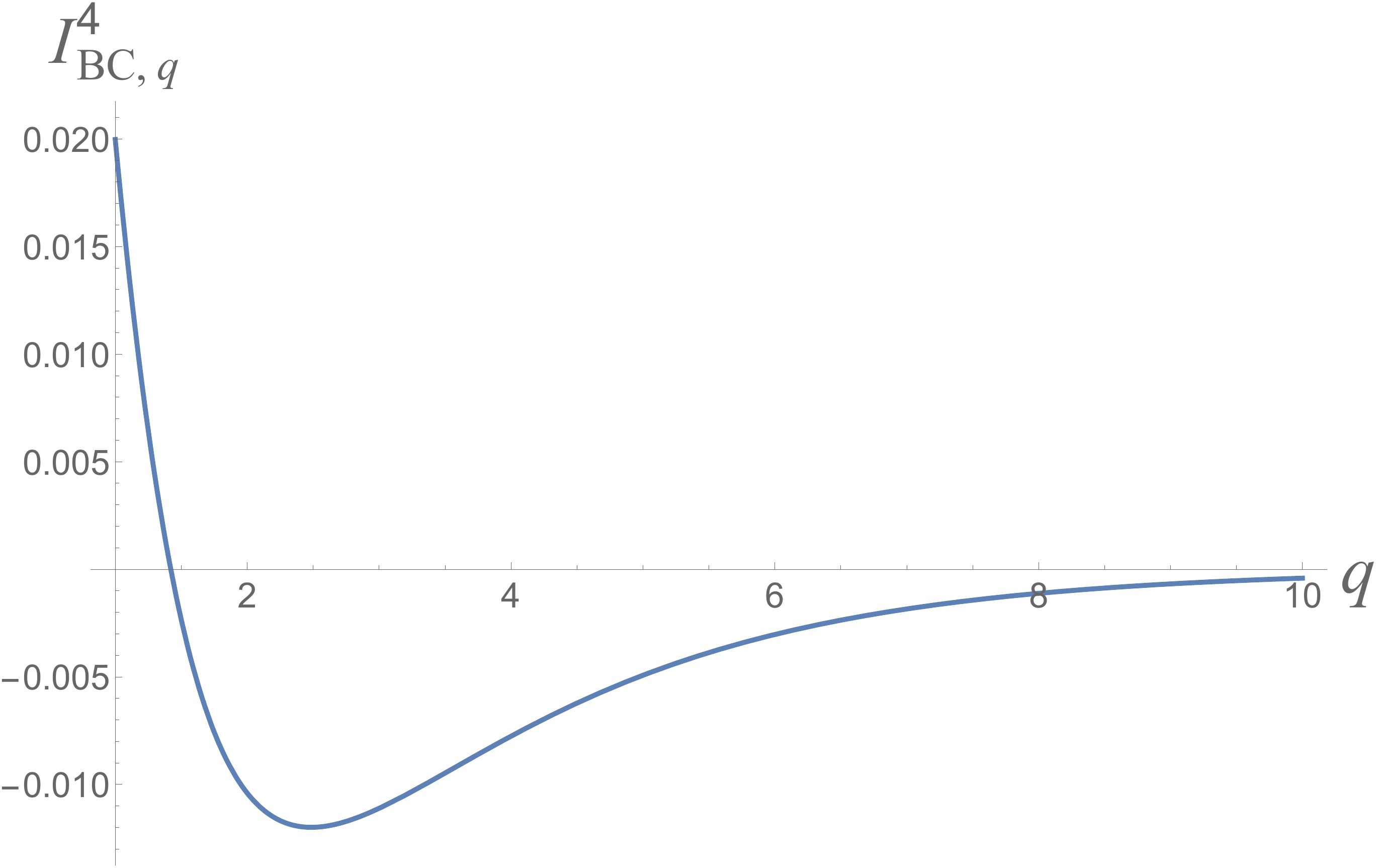}

    \caption{Plot of the $I_{\BC,q}^4$ value as a function of the Tsallis parameter $q$ for the distribution $p_e$ (Equation~\eqref{eq: countereg}). 
      As seen from the plot, the distribution violates the BC inequality $I_{\BC,q}^4\leq 0$ for $q$ values between $1$ (Shannon case) and just under $1.5$ and the violation is maximum in the Shannon case, indicating that it is preferable to use Shannon rather than Tsallis entropies with $q>1$ in this case.}
    \label{fig: countereg_dist}
  \end{figure}
  
\begin{figure}[H]
\centering
\subfloat[]
{
    \includegraphics[scale=.5]{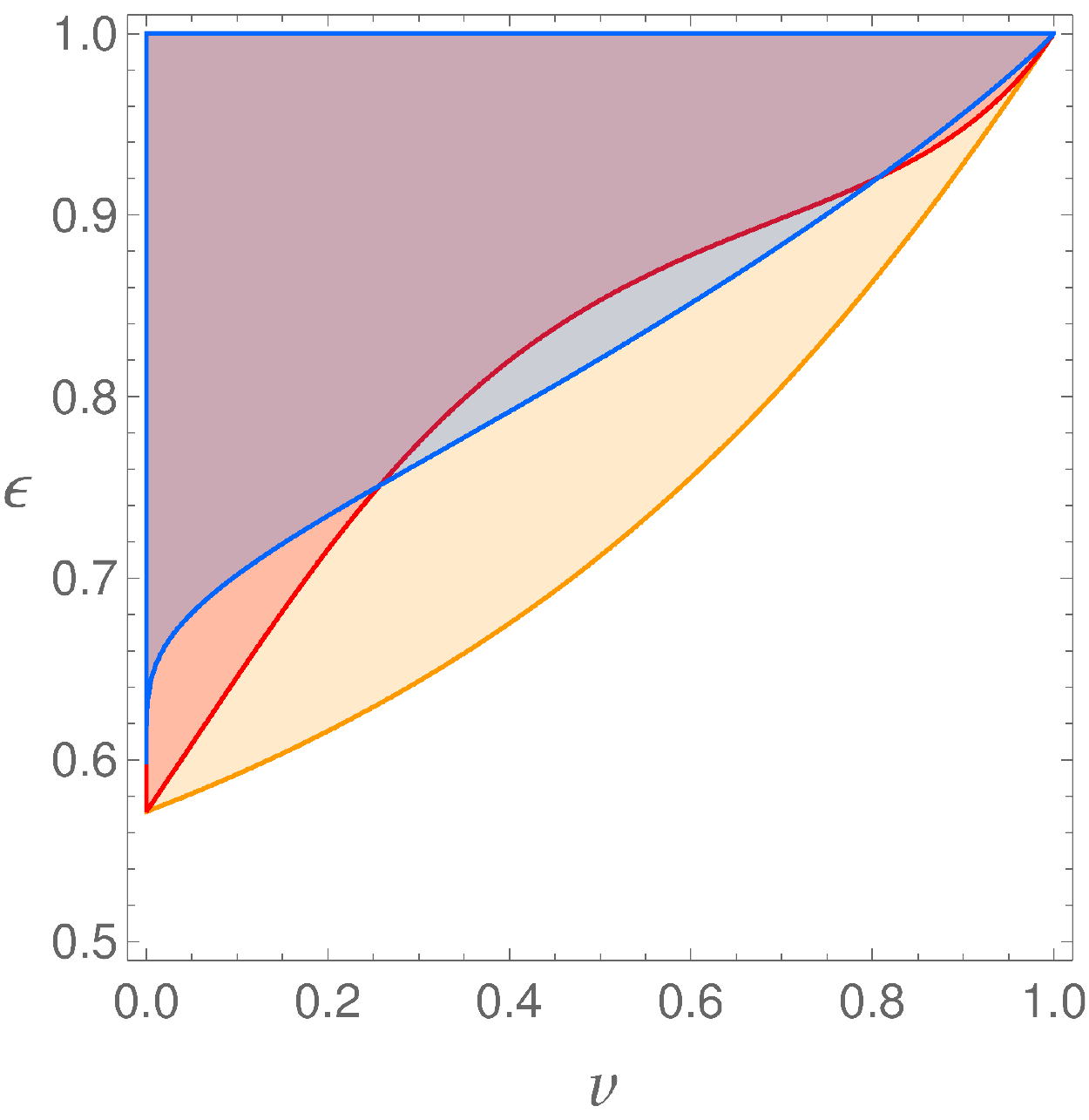}
    \label{fig: 2233shantsal1}
}\qquad
\subfloat[]
{
    \includegraphics[scale=.5]{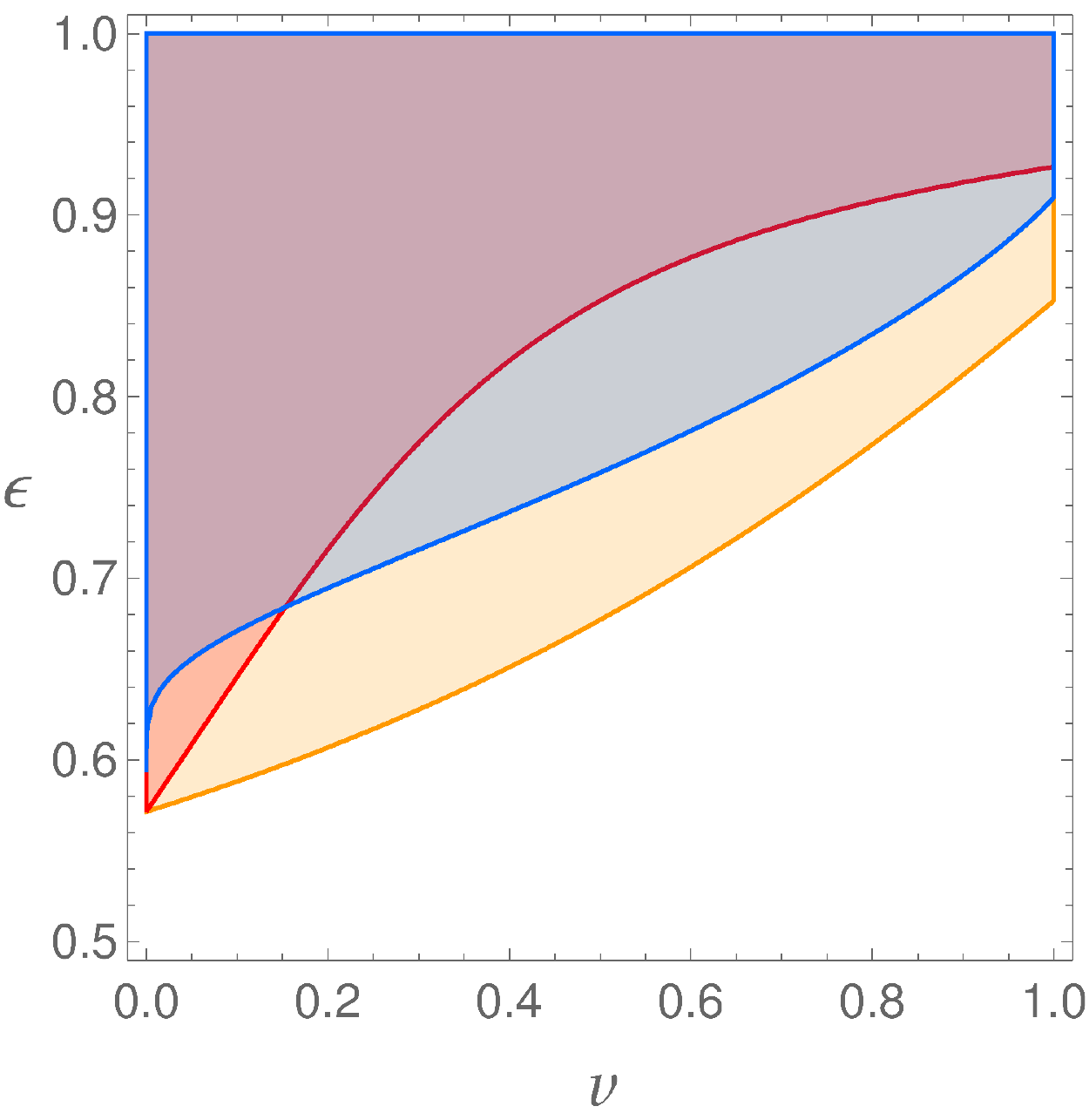}
    \label{fig: 2233shantsal2}
}
 \caption{The regions in the $v-\epsilon$ plane where the Shannon entropic inequality $I_{\BC,1}^4:=I_{\BC}^4 \leq 0$ (blue), the Tsallis entropic inequality, $I_{\BC,q}^4 \leq 0$ for $q=2$ (orange) and $q=8$ (red) are violated by the distributions (a) $p_{\E,\epsilon,v}^{(2,2,3,3)}=vp_{\iso,\epsilon}^{(2,2,3,3)}+(1-v)p_{\Cl}^{(2,2,3,3)}$ (b) $\tilde{p}_{\E,\epsilon,v}^{(2,2,3,3)}=v\tilde{p}_{\iso,\epsilon}^{(2,2,3,3)}+(1-v)p_{\Cl}^{(2,2,3,3)}$. Here $p_{\iso,\epsilon}^{(2,2,3,3)}=\epsilon p_{\NL}+(1-\epsilon)p_{\noise}^{(2,2,3,3)}$ and $\tilde{p}_{\iso,\epsilon}^{(2,2,3,3)}=\epsilon (1/2p_{\NL}+1/2p_{\NL}^*)+(1-\epsilon)p_{\noise}^{(2,2,3,3)}$. For both (a) and (b) $I_{\BC,q}^4\leq 0$ is not violated when $\epsilon \leq 4/7 \approx 0.5714$ but for $\epsilon > 4/7$, there is a violation of this inequality for a larger range of $v$ values in the latter case, and also for a larger range in the $q=2$ case as compared to the other two cases.}
    \label{fig: 2233shantsal}
  \end{figure}

\section{Evidence for Conjectures~\ref{conjecture: 2233shan} and \ref{conjecture: 2233tsal}}
\label{appendix: evidence}
In order to check for violations of the Shannon and Tsallis entropic inequalities $I_{\BC}^4\leq 0$ and $I_{\BC,q}^4\leq 0$ that could be obtained by mixing $p_{\iso, \epsilon}^{(2,2,3,3)}$ (Equation~\eqref{eq: NL3}) with arbitrary classical distributions, we maximized the left hand sides of these inequalities over the polytope $\Conv(\{p_{\iso, \epsilon}^{(2,2,3,3)}\}\bigcup\{p^{\Loc,k}\}_k)$ for some $\epsilon$ values in $(1/2,4/7]$ such as $\epsilon=4/7,5/9$ numerically using {\sc Mathematica}. Note that this polytope contains the local polytope where by definition, entropic inequalities cannot be violated. Thus we can simplify the optimization and increase its reliability by only optimizing over the non-classical part of the polytope $\Conv(\{p_{\iso, \epsilon}^{(2,2,3,3)}\}\bigcup\{p^{\Loc,k}\}_k)$. We find this region as follows. For $1/2<\epsilon\leq 4/7$, we know from Proposition~\ref{prop:nonloc} that $p_{\iso, \epsilon}^{(2,2,3,3)}$ is non-classical but does not violate eny of the CHSH inequalities, while it violates $I_{2233}\leq2$~\eqref{eq: CGLMP}. By Proposition~\ref{proposition: probspace}, this is the only Bell inequality that $p_{\iso, \epsilon}^{(2,2,3,3)}$ violates for this range of $\epsilon$. Thus, the non-classical part of the polytope $\Conv(\{p_{\iso,\epsilon}^{(2,2,3,3)}\}\bigcup\{p^{\Loc,k}\}_k)$ is the convex hull of $p_{\iso, \epsilon}^{(2,2,3,3)}$ and all the local deterministic points that satisfy $I_{2233}=2$. These are the 30 local deterministic points of Table~\ref{tab:vertices}. Hence we only need to optimise over convex combinations of $p_{\iso, \epsilon}^{(2,2,3,3)}$ with these 30 points and not all 81 local deterministic points, which reduces the size of the optimization (number of variables) and increases the chances of it being effective in detecting entropic violations if there are any.

Performing the optimization as outlined above, we found the maximum value to always be non-positive for both Shannon case and the Tsallis case with $q=1.1,2,3,10,50$. We obtained similar results when taking other values of $\epsilon\leq 4/7$ in the distribution $p_{\iso, \epsilon}^{(2,2,3,3)}$ and also when considering the inequalities $I_{\BC,q}^i$ for $i\in \{1,2,3\}$. This suggests that no point in the polytope $\Conv(\{p_{\iso, \epsilon}^{(2,2,3,3)}\}\bigcup\{p^{\Loc,k}\}_k)$ violates any of the (Shannon or Tsallis entropic) BC inequalities for $\epsilon\leq 4/7$. For $\epsilon>4/7$, some mixing of $p_{\iso, \epsilon}^{(2,2,3,3)}$ with $p_{\Cl}^{(2,2,3,3)}$ gives a distribution that violates $I_{\BC,q}^4\leq0$ $\forall q\geq 1$ (cf.\ Proposition~\ref{proposition: BCshan2233}). Note that in the Shannon ($q=1$) case, the range of values of the mixing parameter $v$ for which a violation can be found becomes arbitrarily small as $\epsilon$ approaches $4/7$ from above (see Figure~\ref{fig: 2233shantsal}). This limits the effectiveness of numerical tests for $q$ close to $1$.  For instance, in the Shannon case our program was not able to detect violations of $I^4_{\BC}\leq 0$ for $\epsilon<4.2/7$ (even though our analytic argument shows that these are present), while it was for $\epsilon\geq 4.2/7$. Similarly, for the $q=2$ Tsallis case, violations of $I^4_{\BC,2}\leq 0$ could be found for $\epsilon\geq 4.00001/7$, but not below. The reason for this difference is in line with what one might expect by comparing the plots in Figure~\ref{fig: 2233shantsal}, where for $\epsilon>4/7$ the range of values of $v$ for which a violation is possible is larger in the $q=2$ case. This highlights an advantage of using Tsallis entropy and gives us further confidence that for $1/2\leq\epsilon<4/7$ there is no violation.

However, because of the form of our objective function, the optimisation methods available do not guarantee to find the global maximum. Thus, our findings only constitute evidence for the conjectures and are not conclusive. In general, finding global optima for non-linear, non-convex/concave functions is an open question. A potential avenue for proving these conjectures is using DC (difference of convex) programming~\cite{Horst1999} since our objective function being a linear combination of entropies is a difference of convex functions.

\section{Proofs} \label{appendix: proofs}
For the proofs we need the concept of the local weight of a
non-signalling distribution~\cite{Zukowski99, Cope2019}
\begin{definition}
The \emph{local weight} of a no-signalling distribution $p_{XY|AB}$ is
the largest $\alpha\in[0,1]$ such that we can write
$$p_{XY|AB}=\alpha q^{\Loc}_{XY|AB}+(1-\alpha)q^{\NL}_{XY|AB}\,,$$
where $q^{\Loc}_{XY|AB}$ is an arbitrary local distribution and
$q^{\NL}_{XY|AB}$ is an arbitrary non-signalling distribution.  We
denote the local weight by $l(p_{XY|AB})$.
\end{definition}
The local weight of a distribution can be found by linear programming.

\Propnonloc*

\begin{proof}
The distribution $p_{\iso, \epsilon}^{(2,2,3,3)}=\epsilon p_{\NL}+(1-\epsilon)p_{\noise}^{(2,2,3,3)}$ can be written as follows.
\begin{equation}
\label{eq: pisoAB}
    p_{\iso, \epsilon}^{(2,2,3,3)}=
     \begin{tabular}{ |c|c|} 
 \hline
 $A$ \quad $B$ \quad $B$ & $A$ \quad $B$ \quad $B$\\ 
 $B$ \quad $A$ \quad $B$ & $B$ \quad $A$ \quad $B$\\
 $B$ \quad $B$ \quad $A$ & $B$ \quad $B$ \quad $A$\\
 \hline
$A$ \quad $B$ \quad $B$ & $B$ \quad $A$ \quad $B$\\ 
 $B$ \quad $A$ \quad $B$ & $B$ \quad $B$ \quad $A$\\
 $B$ \quad $B$ \quad $A$ & $A$ \quad $B$ \quad $B$\\
 \hline
\end{tabular}
\end{equation}
where $A=(2\epsilon+1)/9$ and $B=(1-\epsilon)/9$. We used the {\sc LPAssumptions} linear program solver~\cite{LPAssumptions} to find the local weight of $p_{\iso, \epsilon}^{(2,2,3,3)}$, as a function of $\epsilon$ to be
\begin{equation*}
    l(p_{\iso, \epsilon}^{(2,2,3,3)})=\begin{cases}
    1  &0\leq \epsilon \leq \frac{1}{2}\\
    2(1-\epsilon)  &\frac{1}{2}<\epsilon \leq 1
    \end{cases}
\end{equation*}
which establishes the first part of the claim.

The second part can be confirmed by computing the value of each CHSH-type quantity for the distribution $p_{\iso,\epsilon}^{(2,2,3,3)}$ and determining that each has a saturating $\epsilon$ of at most $4/7$.
\end{proof}

\PropBCshan*
\begin{proof}

Consider the function $f:(0,1)\times(0,1)\to\mathbb{R}$ given by
$$f(\epsilon,v):=3(3-2(1-\epsilon)v)\ln[3-2(1-\epsilon)v]+5(1-\epsilon)v\ln[(1-\epsilon)v]-(1+2\epsilon)v\ln[(1+2\epsilon)v]-(3-(2+\epsilon)v)\ln[3-(2+\epsilon)v]-3\ln
9\,,$$
where we implicitly extend the domain to $[0,1]\times[0,1]$ by taking
the relevant limit. The Shannon entropic expression $I_{\BC}^4(\epsilon,v)$ evaluated for the distribution $p_{\E,\epsilon,v}^{(2,2,3,3)}=vp_{\iso, \epsilon}^{(2,2,3,3)}+(1-v)p_{\Cl}^{(2,2,3,3)}$ (seen as a function of $\epsilon$ and $v$) is then given as $$I_{\BC}^4=\frac{1}{3\ln[2]}f(\epsilon,v)$$
Thus all the following arguments for $f(\epsilon,v)$ also hold for $I_{\BC}^4$. 
\par
  We first use that for $c>0$ and $a\in\mathbb{R}$ for sufficiently
  small $v$ we have
  $$\ln[c+av]=\ln[c]+\frac{av}{c}+O(v^2)\,.$$
Using this we can expand $f(\epsilon,v)$ for small $v$ as
\begin{equation}\label{eq:sm}
  f(\epsilon,v)=(4-7\epsilon)v\ln[v]-(4-7\epsilon)v(1+\ln[3])+v(5(1-\epsilon)\ln[1-\epsilon]-(1+2\epsilon)\ln[1+2\epsilon])+O(v^2)\,.
\end{equation}
Thus, since $\lim_{v\to0}v\ln[v]=0$ we have $\lim_{v\to0}f[\epsilon,v]=0$.

We also have
\begin{align}
  \frac{\partial}{\partial v}f(\epsilon,v)&=(4-7\epsilon)\ln\left[v\right]+5(1-\epsilon)\ln[1-\epsilon]-(1+2\epsilon)\ln[1+2\epsilon]-6(1-\epsilon)\ln[3-2(1-\epsilon)v]+(2+\epsilon)\ln[3-(2+\epsilon)v]\,.\label{eq:deriv}
\end{align}
Note that $5(1-\epsilon)\ln[1-\epsilon]\leq 0$,
$-(1+2\epsilon)\ln[1+2\epsilon]\leq 0$ $\forall \epsilon \in
[0,1]$. Further, since $3-2(1-\epsilon)v\geq 3-(2+\epsilon)v$ and both
terms are positive, $6(1-\epsilon)\geq(2+\epsilon)$ $\forall \epsilon
<4/7$, and using the fact that $\ln[]$ is an increasing function, we have $-6(1-\epsilon)\ln[3-2(1-\epsilon)v]+(2+\epsilon)\ln[3-(2+\epsilon)v]\leq 0$ $\forall \epsilon \in [0,4/7]$, $v\in [0,1]$. This in turn implies that
\begin{align}
  \frac{\partial}{\partial v}f(\epsilon,v)\leq(4-7\epsilon)\ln\left[v\right]\qquad \forall\ 0\leq \epsilon \leq 4/7, 0\leq v \leq 1 \nonumber
\end{align}
Hence we can conclude that for $\epsilon\leq4/7$,
$ \frac{\partial}{\partial v}f(\epsilon,v)<0$ for all $v\in [0,1]$. Thus, $f(\epsilon,v)$ is zero at $v=0$ and, for $\epsilon\leq4/7$, decreases
with $v$, implying that $f(\epsilon,v)\leq 0$ $\forall \epsilon\in
[0,4/7]$, $v\in [0,1]$. Note that $p_{\E,\epsilon,v}^{(2,2,3,3)}=vp_{\iso,
  \epsilon}^{(2,2,3,3)}+(1-v)p_C^{(2,2,3,3)}$ does not violate any of
the analogous inequalities  $I_{\BC}^i\leq 0$ 
for any $i\in \{1,2,3\}$, $\epsilon,v  \in [0,1]$.
This is because for this distribution, we always have
$H(X_0Y_0)=H(X_0Y_1)=H(X_1Y_0)$, $H(X_0)=H(X_1)=H(Y_0)=H(Y_1)$. Thus
all three inequalities $I_{\BC}^1\leq 0$, $I_{\BC}^2\leq 0$ and $I_{\BC}^3\leq 0$ reduce to
$H(X_1)+H(Y_1)-H(X_1Y_0)-H(X_1Y_1)\leq 0$, which is always satisfied since $H(X_1)\leq H(X_1Y_0)$ and $H(Y_1)\leq H(X_1Y_1)$ by the monotonicity of Shannon entropy.

Further, using the expression for the derivative of $f(\epsilon,v)$
with respect to $v$ in Equation~\eqref{eq:deriv}, we find that for
$\epsilon>4/7$,
$\lim_{v\to0} \frac{\partial}{\partial v}f(\epsilon,v)=\infty$. Thus,
since $f(\epsilon,v)=0$ for $v=0$, sufficiently close to $v=0$ there
exists a $v$ such that $f(\epsilon,v)>0$. This proves the claim.
\end{proof}

\PropBCtsal*
\begin{proof}
The Tsallis entropic expression $I_{\BC,q}^4(\epsilon,v)$ evaluated for the distribution $p_{\E,\epsilon,v}^{(2,2,3,3)}=vp_{\iso, \epsilon}^{(2,2,3,3)}+(1-v)p_{\Cl}^{(2,2,3,3)}$ (seen as a function of $q$, $\epsilon$ and $v$) is given as
$$I_{\BC,q}^4=\frac{1}{q-1}\Bigg[9\Bigg(\frac{3-2(1-\epsilon)v}{9}\Bigg)^q+15\Bigg(\frac{(1-\epsilon)v}{9}\Bigg)^q-\frac{6}{3^q}-3\Bigg(\frac{3-(2+\epsilon)v}{9}\Bigg)^q-3\Bigg(\frac{(1+2\epsilon)v}{9}\Bigg)^q\Bigg]=:\frac{g(q,\epsilon,v)}{q-1}$$
For $q>1$, the following arguments for $g(q,\epsilon,v)$ also hold for
$I_{\BC,q}^4$. Note that
\begin{equation}
\label{eq: derivTsal}
  \frac{\partial}{\partial v}g(q,\epsilon,v)=\frac{q}{3^{2q-1}}\Big[-6(1-\epsilon)\Big(3-2(1-\epsilon)v\Big)^{q-1}+5(1-\epsilon)\Big((1-\epsilon)v\Big)^{q-1}+(2+\epsilon)\Big(3-(2+\epsilon)v\Big)^{q-1}-(1+2\epsilon)\Big((1+2\epsilon)v\Big)^{q-1}\Big] \,.
\end{equation}

Then, since
\begin{align*}
  6(1-\epsilon)\Big(3-2(1-\epsilon)v\Big)^{q-1}\geq (2+\epsilon)\Big(3-(2+\epsilon)v\Big)^{q-1}\qquad \text{and}\qquad  
  (1+2\epsilon)\Big((1+2\epsilon)v\Big)^{q-1}\geq 5(1-\epsilon)\Big((1-\epsilon)v\Big)^{q-1}
\end{align*}
hold for all $\epsilon\leq4/7$, $v\in[0,1]$ and $q>1$, we have 
$$\frac{\partial}{\partial v}g(q,\epsilon,v)\leq 0 \qquad \forall \epsilon\leq 4/7,\ v\in [0,1],\ q> 1.$$
Since $g(q,\epsilon,v=0)=0$, this implies that
$g(q,\epsilon,v)\leq 0$ $\forall \epsilon\leq 4/7, v\in [0,1], q>
1$. Hence, for $\epsilon\leq 4/7$ we cannot violate $I_{\BC,q}^4\leq0$
for any $v\in[0,1]$, $q>1$.

Analogously to the Shannon case, $p_{\E,\epsilon,v}^{(2,2,3,3)}=vp_{\iso,
  \epsilon}^{(2,2,3,3)}+(1-v)p_{\Cl}^{(2,2,3,3)}$ also does not
violate any of the inequalities $I_{\BC,q}^1\leq0$, $I_{\BC,q}^2\leq0$
or $I_{\BC,q}^3\leq0$ for any $\epsilon,v\in[0,1]$ and $q>1$ by the
same argument as in Proposition~\ref{proposition: BCshan2233}.

Further, Equation~\eqref{eq: derivTsal} implies $\lim_{v\to0}\frac{\partial}{\partial v}g(q,\epsilon,v)=\frac{q}{3^q}(7\epsilon-4)$ which is always positive for $\epsilon>4/7$. Since $g(q,\epsilon,v)=0$ for $v=0$, this allows us to conclude that for $\epsilon>4/7$, there
  exists a $v$ sufficiently close to $v=0$ such that $g(q,\epsilon,v)>0$. This establishes the claim.
\end{proof}

\PropCGs*
\begin{proof}
For the ``if'' part of the proof, we calculated all the 255
coarse-grainings of $p_{\iso, \epsilon=4/7}^{(2,2,3,3)}$ and used a
linear programming algorithm to find that all of these are local
(their local weight equals $1$). Since decreasing $\epsilon$ in $p_{\iso,
  \epsilon}^{(2,2,3,3)}$ cannot increase the violation of any (probability space) Bell inequality and neither can coarse-grainings, it follows that if $\epsilon\leq 4/7$, all coarse-grainings of $p_{\iso, \epsilon}^{(2,2,3,3)}$ are classical.

For the ``only if'' part we need to show that if all coarse-grainings of $p_{\iso, \epsilon}^{(2,2,3,3)}$ are classical then $\epsilon\leq 4/7$ or equivalently, if $\epsilon>4/7$, there exists at least one coarse-graining of $p_{\iso, \epsilon}^{(2,2,3,3)}$ that is non-classical. Consider the coarse-graining that involves combining the second output with the first for all 4 input choices. For $p_{\iso, \epsilon}^{(2,2,3,3)}$ as in Equation~\eqref{eq: pisoAB}, this coarse-graining gives
\begin{equation}
\label{eq: pisoCG}
\setlength{\tabcolsep}{0.7em}
    p_{\CG, \epsilon}^{(2,2,3,3)}=
     \begin{tabular}{ |ccc|ccc|} 
 \hline
 $2(A+B)$ & 0 & $2B$ & $2(A+B)$ & 0 & $2B$\\ 
 0 & 0 & 0 & 0 & 0 & 0\\
 $2B$ & 0 & $A$ & $2B$ & 0 & $A$\\
 \hline
 $2(A+B)$ & 0 & $2B$ & $3B+A$ & 0 & $A+B$\\ 
 0 & 0 & 0 & 0 & 0 & 0\\
 $2B$ & 0 & $A$ & $A+B$ & 0 & $B$\\
 \hline
\end{tabular}\,,
\end{equation}
where $A=(2\epsilon+1)/9$ and $B=(1-\epsilon)/9$.  The $I_{2233}$ value or left hand side of Equation~\eqref{eq: CGLMP} for this distribution is $9A-3B$. For this to be classical, we require that $9A-3B\leq 2$ which gives $\epsilon\leq 4/7$. Again using the {\sc LPAssumptions} linear program solver~\cite{LPAssumptions} we found the local weight of this distribution as a function of $\epsilon$, which gives the following.
\begin{equation*}
    l(p_{\CG, \epsilon}^{(2,2,3,3)})=\begin{cases}
    1,  &0\leq \epsilon \leq \frac{4}{7}.\\
    \frac{1}{9}(17-14\epsilon),  &\frac{4}{7}<\epsilon \leq 1.
    \end{cases}
\end{equation*}
In other words, if $\epsilon>4/7$, then the coarse-graining $p_{\CG,\epsilon}^{(2,2,3,3)}$ of $p_{\iso,\epsilon}^{(2,2,3,3)}$ violates the $I_{2233}$ Inequality~(\ref{eq: CGLMP}) and is hence non-classical. This concludes the proof.
\end{proof}

We prove the following two propositions together as they only differ in one step.
\ProprelabsA*
\ProprelabsB*
\begin{proof}
  If Conjectures~\ref{conjecture: 2233shan} and~\ref{conjecture: 2233tsal} hold, then for any $\epsilon\leq4/7$, $I_{\BC,q}^i\leq 0$ holds $\forall q\geq 1$, $\forall i \in \{1,2,3,4\}$ and for all distributions in the polytope $\Conv(\{p_{\iso, \epsilon}^{(2,2,3,3)}\}\bigcup\{p^{\Loc,k}\}_k\})$.\footnote{Note that $q=1$ covers the Shannon case.} We want to show that this implies the same for the larger polytope that comprises the convex hull of not just $p_{\iso, \epsilon}^{(2,2,3,3)}$ and local deterministic distributions $\{p^{\Loc,k}\}_k$, but also all the local relabellings of $p_{\iso, \epsilon}^{(2,2,3,3)}$, i.e., the polytope $\Pi_{\epsilon}=\Conv(\{p^{\RL,j}_\epsilon\}_j\bigcup\{p^{\Loc,k}\}_k)$. While we considered $p_{\iso, \epsilon}^{(2,2,3,3)}$ in Propositions~\ref{proposition: BCshan2233} and \ref{proposition: BCtsal2233} and Conjectures~\ref{conjecture: 2233shan} and~\ref{conjecture: 2233tsal}, due to symmetries, these also apply to every relabelling of $p_{\iso, \epsilon}^{(2,2,3,3)}$, i.e., $I_{\BC,q}^i\leq 0$ ($\forall i\in \{1,2,3,4\}$) throughout every polytope in the set $\{\Conv(\{p^{\RL,j}_\epsilon\}\bigcup\{p^{\Loc,k}\}_k)\}_j$ (where $j$ runs over all the local relabellings of $p_{\iso, \epsilon}^{(2,2,3,3)}$)\footnote{A note on notation: $\{p^{\RL,j}_\epsilon\}$ is a set comprising a single element, while $\{p^{\RL,j}_\epsilon\}_j$ is a set whose elements are the distributions for every $j$.}. This is because for every input-output relabelling of $p_{\iso, \epsilon}^{(2,2,3,3)}$, we can correspondingly relabel the inequality expression $I_{\BC,q}^4$ (for $q\geq 1$) and the same arguments as Propositions~\ref{proposition: BCshan2233} and \ref{proposition: BCtsal2233} again hold, and similarly Conjectures~\ref{conjecture: 2233shan} and~\ref{conjecture: 2233tsal} also follow for this case.\footnote{Note that output relabellings don't change the entropic expression but input relabellings (4 in number) can give either one of $I_{\BC,q}^i\leq 0$ for $i\in \{1,2,3,4\}$. Thus for an output relabelling of $p_{\iso, \epsilon}^{(2,2,3,3)}$, one can continue using $I_{\BC,q}^4$ in Propositions~\ref{proposition: BCshan2233} and \ref{proposition: BCtsal2233} and the following Conjectures while for input relabellings, one simply needs to relabel the inequalities accordingly and run the same arguments.} From this argument, it follows that if Conjectures~\ref{conjecture: 2233shan} and~\ref{conjecture: 2233tsal} hold, then $I_{\BC,q}^i\leq 0$ $\forall q \geq 1$, $\forall i \in \{1,2,3,4\}$ everywhere in the union of the polytopes, i.e., everywhere in $\bigcup_j\Conv(\{p^{\RL,j}_\epsilon\}\bigcup\{p^{\Loc,k}\}_k)$.

  To conclude the proof it remains to show that $\bigcup_j\Conv(\{p^{\RL,j}_\epsilon\}\bigcup\{p^{\Loc,k}\}_k)=\Pi_\epsilon$ $\forall \epsilon \leq 4/7$. This is established below. Then, Proposition~\ref{proposition: relab2} automatically follows while Proposition~\ref{proposition: relab1} follows from this and Lemma~\ref{lem:BCcomplete}.
\end{proof}

\begin{proposition}
\label{proposition: R1}
$p^j_{\mix,\epsilon}:= \frac{1}{2}p_{\iso, \epsilon}^{(2,2,3,3)}+\frac{1}{2}p^{\RL,j}_\epsilon$ is local $\forall j\neq 1$ if and only if $\epsilon \leq 4/7$.
\end{proposition}
\begin{proof}
For the ``if'' part of the proof, we used a linear program to confirm that $p^j_{\mix,\epsilon=4/7}$ is local\footnote{$j=1$ is excluded since $p^1_{\mix,\epsilon}=p_{\iso, \epsilon}^{(2,2,3,3)}$ which is non-classical for $\epsilon > 1/2$.} $\forall j\neq 1$. Since reducing $\epsilon$ in $p^j_{\mix,\epsilon}$ cannot decrease the local weight, this also holds for $\epsilon<4/7$.

The ``only if'' part of the proof is equivalent to showing that $\forall \epsilon > 4/7$, $\exists j$ such that $p^j_{\mix,\epsilon}$ is non-classical. Consider the particular local relabelling that corresponds to Alice swapping the outputs ``1'' and ``2'' only when her input is $A=1$. Let the distribution obtained by applying this relabelling to $p_{\iso, \epsilon}^{(2,2,3,3)}$ be $p^{\RL}_\epsilon$. Then $p_{\mix,\epsilon}=\frac{1}{2}p_{\iso, \epsilon}^{(2,2,3,3)}+\frac{1}{2}p^{\RL}_\epsilon$. More explicitly, 
\begin{equation}
    p^{\RL}_\epsilon=
     \begin{array}{|c|c|} 
 \hline
 A \quad B \quad B & A \quad B \quad B\\ 
 B \quad A \quad B & B \quad A \quad B\\
 B \quad B \quad A & B \quad B \quad A\\
 \hline
A \quad B \quad B & B \quad A \quad B\\ 
 B \quad B \quad A & A \quad B \quad B\\
  B \quad A \quad B & B \quad B \quad A\\
 \hline
\end{array}\qquad \text{and} \qquad p_{\mix,\epsilon}=
     \begin{array}{ |c|c|} 
 \hline
 A \quad B \quad B & A \quad B \quad B\\ 
 B \quad A \quad B & B \quad A \quad B\\
 B \quad B \quad A & B \quad B \quad A\\
 \hline
A \quad B \quad B & B \quad A \quad B\\ 
 B \quad * \quad * & * \quad B \quad *\\
 B \quad * \quad * & * \quad B \quad *\\
 \hline
\end{array}\quad,
\end{equation}
where $*=\frac{A+B}{2}$, $A=(2\epsilon+1)/9$ and $B=(1-\epsilon)/9$. Now consider the Bell inequality $\Tr(M^TP)\geq 1$ where $P$ is an arbitrary distribution in the $(2,2,3,3)$ scenario and $M$ is the following matrix.
\begin{equation}
    M:=
     \begin{array}{ |c|c|} 
 \hline
 0 \quad 1 \quad 1 & 0 \quad 1 \quad 1\\ 
 1 \quad 0 \quad 1 & 1 \quad 0 \quad 1\\
 1 \quad 1 \quad 1 & 0 \quad 0 \quad 0\\
 \hline
0 \quad 1 \quad 0 & 1 \quad 0 \quad 0\\ 
 1 \quad 0 \quad 0 & 0 \quad 1 \quad 0\\
 1 \quad 0 \quad 0 & 0 \quad 1 \quad 0\\
 \hline
\end{array}
\end{equation}
Then, the condition for $p_{\mix,\epsilon}$ to be non-classical with respect to this Bell inequality i.e., $\Tr(M^Tp_{\mix,\epsilon})<1$ gives us $\epsilon>4/7$. Since $p^{\RL}_\epsilon$ is a local relabelling of $p_{\iso, \epsilon}^{(2,2,3,3)}$, there exists a $j$ such that $p^{\RL}_\epsilon=p^{\RL,j}_\epsilon$ and hence $p_{\mix,\epsilon}=p^j_{\mix,\epsilon}$. Thus we have shown that whenever $\epsilon>4/7$, $\exists j$ such that $p^j_{\mix,\epsilon}$ is non-classical and hence $p^j_{\mix,\epsilon}$ is local for all $j$ implies that $\epsilon \leq 4/7$ which concludes the proof.
\end{proof}

By symmetry, there is an analogue of Proposition~\ref{proposition: R1} with $p_{\iso, \epsilon}^{(2,2,3,3)}$ replaced by $p^{\RL,i}_\epsilon$ for any $i$, so we have the following corollary.
\begin{corollary}
\label{corollary: R1}
$\tilde{p}^{i,j}_{\mix,\epsilon}:= \frac{1}{2}p^{\RL,i}_\epsilon+\frac{1}{2}p^{\RL,j}_\epsilon$ is local $\forall j\neq i$ if and only if $\epsilon \leq 4/7$.
\end{corollary}

\begin{theorem}[Bemporad et.\ al 2001 \cite{Bemporad2001}]
\label{theorem: bemporad}
Let $\mathcal{P}$ and $\mathcal{Q}$ be polytopes with vertex sets $V$ and $W$ respectively i.e., $\mathcal{P}=\Conv(V)$ and $\mathcal{Q}=\Conv(W)$. Then $\mathcal{P}\bigcup\mathcal{Q}$ is convex if and only if the line-segment $[v,w]$ is contained in $\mathcal{P}\bigcup\mathcal{Q}$ $\forall v\in V$ and $w \in W$.
\end{theorem}

Let $\cP_j=\Conv(\{p^{\RL,j}_\epsilon\}\bigcup\{p^{\Loc,k}\}_k\})$ and $\mathbb{P}$ be the set of polytopes $\mathbb{P}:=\{\cP_j\}_j$. We use the above theorem to prove the final result that establishes Propositions~\ref{proposition: relab1} and~\ref{proposition: relab2}.
\begin{lemma}
Let $V_j$ be the vertex set of the polytope $\cP_j \in \mathbb{P}$ and $V:=\bigcup_j V_j$. Then, $\bigcup_{\cP_j\in \mathbb{P}}\cP_j=\Conv(\bigcup_i V_j)=\Conv(V)=\Pi_\epsilon$ $\forall \epsilon \leq 4/7$.
\end{lemma}
\begin{proof}
  By Corollary~\ref{corollary: R1}, for $i\neq j$ we have that $\frac{1}{2}\left(p^{\RL,i}_\epsilon+p^{\RL,j}_\epsilon\right) \in \mathcal{L}^{(2,2,3,3)}=\cP_i\bigcap\cP_j$ $\forall \epsilon \leq 4/7$. This implies that $\alpha p^{\RL,i}_\epsilon+(1-\alpha)p^{\RL,j}_\epsilon\in \cP_i\bigcup\cP_j$ $\forall \epsilon \leq 4/7$, $\alpha\in [0,1]$, i.e., the line segment $[p^{\RL,i}_\epsilon,p^{\RL,j}_\epsilon]$ is completely contained in the union of the corresponding polytopes $\cP_i\bigcup\cP_j$. Note that all other line segments $[v_i,v_j]$ with $v_i \in V_i$ and $v_j \in V_j$ are contained in $\cP_i\bigcup\cP_j$ by construction since at least one of $v_i$ or $v_j$ would be a local-deterministic vertex. Therefore, by Theorem~\ref{theorem: bemporad}, $\cP_i\bigcup\cP_j$ is convex $\forall i,j$ and $\epsilon \leq 4/7$. We can then apply Proposition~\ref{proposition: R1} and Theorem~\ref{theorem: bemporad} to the convex polytopes $\cP_i\bigcup\cP_j$ and $\cP_k$ and show that $\cP_i\bigcup\cP_j\bigcup\cP_k$ is convex $\forall i,j,k$ and $\epsilon \leq 4/7$. Proceeding in this way, we conclude that $\bigcup_{\cP_i\in \mathbb{P}}\cP_i$ is convex $\forall \epsilon \leq 4/7$, and hence $\bigcup_{\cP_j\in \mathbb{P}} \cP_j=\Conv(\bigcup_i V_j)=\Conv(V)=\Pi_\epsilon$.
\end{proof}

\end{document}